%% file: MeanStatOPE.tex
\documentclass{ectjEDIT}

\input{def.tex}

\received{February 2021}

\title[Adaptive Doubly Robust Estimator from Non-Stationary Logging Policy]{Adaptive Doubly Robust Estimator from Non-stationary Logging Policy under a Convergence of Average Probability}
                        
\author{Masahiro Kato}

\address{Cyberagent, Inc.\\
Shibuya, Tokyo}
\email{masahiro\_kato@cyberagent.co.jp}

\def\AmSTeX{$\cal A$\kern-.1667em\lower.5ex\hbox{$\cal M$}\kern-.125em
            $\cal S$-\TeX}
\def\BibTeX{{\rm B\kern-.05em{\sc i\kern-.025em b}\kern-.08em
            T\kern-.1667em\lower.7ex\hbox{E}\kern-.125emX}}

\begin{document}

\begin{abstract}
\emph{Adaptive experiments}, including efficient average treatment effect estimation and multi-armed bandit algorithms, have garnered attention in various applications, such as social experiments, clinical trials, and online advertisement optimization. This paper considers estimating the \emph{mean outcome} of an action from samples obtained in adaptive experiments. In causal inference, the mean outcome of an action has a crucial role, and the estimation is an essential task, where the average treatment effect estimation and off-policy value estimation are its variants. In adaptive experiments, the probability of choosing an action (logging policy) is allowed to be sequentially updated based on past observations. Due to this logging policy depending on the past observations, the samples are often not \emph{independent and identically distributed} (i.i.d.), making obtaining an asymptotically normal estimator difficult. A typical approach for this problem is to assume that the logging policy converges in a time-invariant function. However, this assumption is restrictive in various applications, such as when the logging policy fluctuates or becomes zero at some periods. To mitigate this limitation, we propose another assumption that \emph{the average logging policy} converges to a time-invariant function and show the doubly robust (DR) estimator's asymptotic normality. Under the assumption, the logging policy itself can fluctuate or be zero for some actions. We also show the empirical properties by simulations.
\end{abstract}

\section{Introduction}
Estimating the \emph{mean outcome} of an action is an essential task in statistical inference under Neyman-Rubin potential outcomes model \citep{Luedtke2016}. The average treatment effect (ATE) estimation \citep{Holland1986stat,rubin87,robins94,hirano2003efficient,ImaiKosuke2014Cbps,imbens_rubin_2015} and off-policy value (OPV) estimation for multi-armed bandit (MAB) algorithms \citep{Precup2000,dudik2011doubly,Mahmood20014,Li2015,jiang2016,wang2017optimal,Bibaut2019moreffficient} are its special cases. We consider mean outcome estimation for dependent samples obtained from adaptive experiments \citep{ChowChang201112,Hahn2011,Kasy2021}, including MAB algorithms \citep{Villar2018} and treatment regimes (TR) \citep{Zhang2012,zhao2012est_ind,Chakraborty2013}. In adaptive experiments, we gather samples via \emph{logging policy} (probability of choosing an action), which is sequentially updated based on past observations. For instance, \citet{Laan2008TheCA} considered a situation where a research subject visits at each period $t=1,2,\dots,T$, and we select a treatment following a logging policy sequentially updated based on past observation to minimize the variance of an ATE estimator. Owing to the logging policy depending on the past observations, the samples are not \emph{independent and identically distributed} (i.i.d.). 

For statistical inference for the mean outcome, we aim to construct an asymptotically normal estimator, which also implies $\sqrt{T}$ consistency for a sample size $T$. Under the dependency, we cannot apply the standard central limit theorem (CLT). For mitigating this problem, existing studies proposed various approaches, and one of the main approaches is to apply the martingale CLT, which requires that the variance of the target random variable to a time-invariant one. Existing studies have proposed the following three strategies for satisfying this requirement. The first strategy is to assume that the logging policy converges to a time-invariant function in probability \citep{Laan2008TheCA,hadad2019,Kato2020adaptive}. The second strategy is to assume the existence of batched samples, where there are infinite samples in each batch \citep{Hahn2011,Laan2014onlinetml,kelly2020batched,kato2020batch}. The third strategy is to standardize the score function to equalize the variance of each period \citep{Luedtke2016}. 

\begin{figure}[t]
\begin{center}
 \includegraphics[width=90mm]{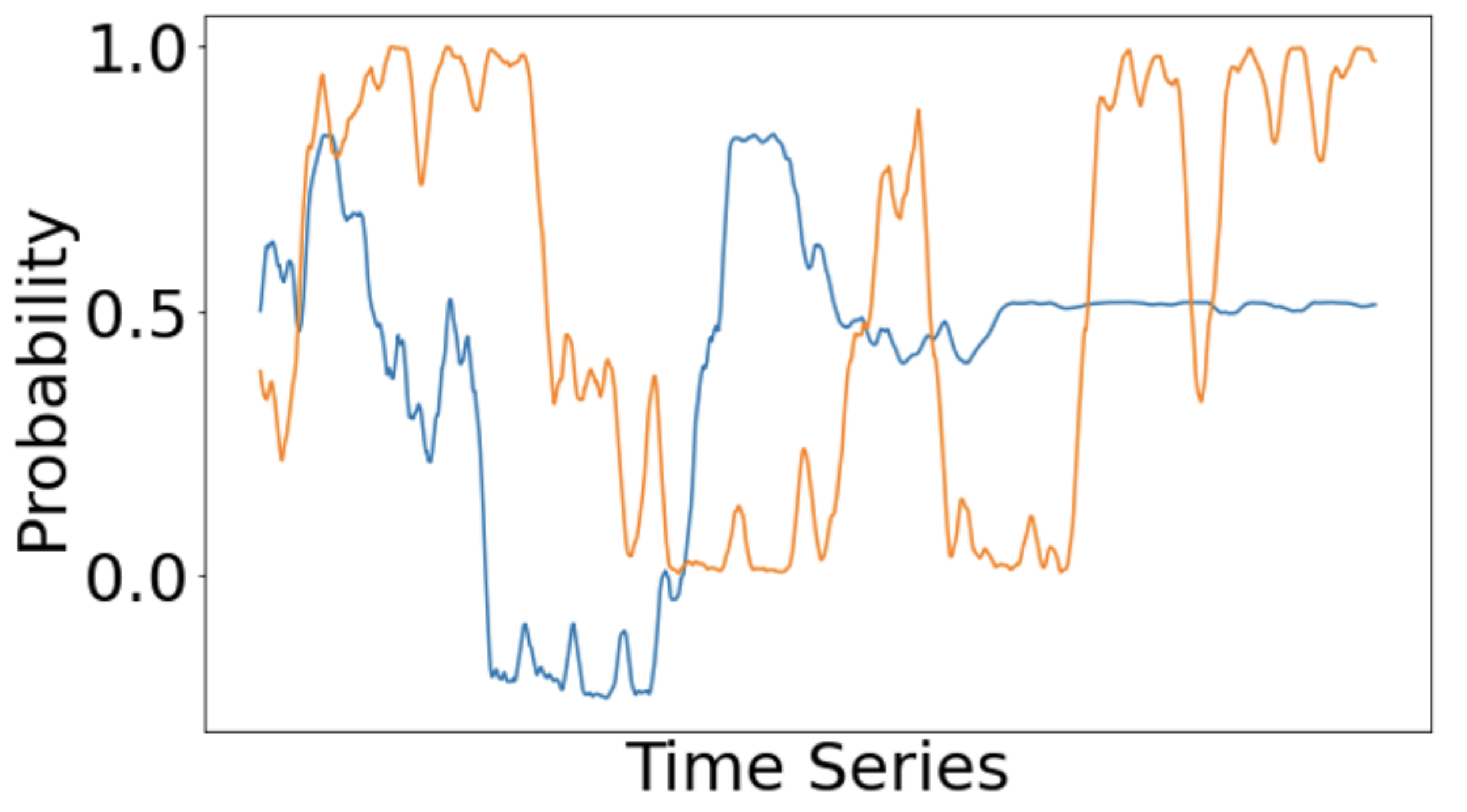}
\end{center}
\caption{The blue line denotes the logging policy converging to $0.5$. The orange line denotes the logging policy that does not converge, but the average converges to $0.5$. Existing studies mainly considered mean outcome estimation from dependent samples for a converging logging policy, such as the blue line. This paper shows how to obtain an asymptotically normal estimator from a logging policy without its convergence. A fundamental assumption is an average convergence as represented by the orange line.}
\label{fig:concept}
\end{figure}

However, these strategies are often restrictive in practice. For instance, we can raise the following two situations, where the first strategy is not applicable: (I) the logging policy fluctuates, and (II) the logging policy is $0$ or $1$ at a period. In best arm identification (BAI), \citet{kaufman2016complexity} and \citet{garivier2016optimal} showed that pulling arms with a specific ratio achieves the lower bound of the sample complexity. Their methods deterministically select an arm at a period to keep the ratio. Here, the logging policy does not converge to a time-invariance function. The value is $1$ for an arm and $0$ for the others. Therefore, we cannot apply an existing mean outcome estimator for the situations.

When the logging policy fluctuates, the martingale CLT is not applicable owing to the time-variant variance. To mitigate this problem, \citet{Luedtke2016} proposed standardizing the score function by its estimated variance, which is applicable to many cases. However, the estimator of \citet{Luedtke2016} does not achieve $\sqrt{T}$-consistency by splitting the samples to estimate the variance for the standardization. 

For overcoming these problems, instead of the conventional strategies, this paper proposes a new strategy based on the assumption that \emph{the average logging policy converges to a time-invariant function in probability}. This assumption is a generalization of the first strategy based on the assumption that the logging policy converges to a time-invariant function in probability. This is because when the logging policy itself converges, the average logging policy also converges. The new assumption is greatly useful in practice. For instance, we can apply our method to cases, where the logging policy can fluctuate, and the logging policy is $0$ or $1$  if the average logging policy converges. We also illustrate an example when the average logging policy converges in contrast to a case when the logging policy itself converges in Figure~\ref{fig:concept}.

\paragraph{Organization of this paper.} In Section~\ref{sec:problem}, we introduce our problem setting and the parameter that we want to estimate. In Section~\ref{sec:adr}, we propose two DR-type estimators and show the asymptotic normalities under the assumption that the average logging policy converges in probability. The first estimator is more natural and empirically performs well but requires conditions that are not easy to be confirmed (Theorem~\ref{thm:asymp_normal}). The second estimator does not empirically perform well as the first one, but we can show the asymptotic normality by using an assumption that is easier to be confirmed (Theorem~\ref{thm:asymp_normal2}). In Section~\ref{sec:exp}, we numerically investigate the performance of the proposed estimators. In Section~\ref{sec:discuss}, we discuss the remaining problems.

\section{Problem Setting}
\label{sec:problem}
In this section, we describe our problem setting.

\subsection{Data-Generating Process (DGP)}
Consider a time-series $t=1,2,\dots, T$. For each period $t$, let $A_t$ be an \emph{action} in $\mathcal{A}=\{1,2,\dots,K\}$, $X_t\in\mathcal{X}$ be a \emph{covariate} observed by the decision maker when choosing an action, and $\mathcal{X}$ be the space of covariate. Let a random variable denoting an outcome at period $t$ be $Y_t=\sum^K_{a=1}\mathbbm{1}[A_t = a]Y_t(a)$, where $Y_t(a)\in\mathbb{R}$ is a random variable denoting the potential (random) outcome of an action $a\in\mathcal{A}$. We have a dataset $\big\{(X_t, A_t, Y_t)\big\}^{T}_{t=1}$. The DGP is described as follows:
\begin{align*}
(X_t, A_t, Y_t(A_t)) \sim P_t = p(x)p_t(a\mid x)p(y_a\mid x),
\end{align*}
where $X_t$ is generated from $p(x)$, $A_t$ is generated from $p_t(a\mid x)$ at the period $t$, and $Y_t(a)$ is generated from $p(y_a\mid x)$. While $p(x)$ and $p(y_a\mid x)$ are invariant across periods, $p_t(a\mid x)$ can be different across periods based on past observations. In this case, the samples $\big\{(X_t, A_t, Y_t)\big\}^{T}_{t=1}$ are correlated over time; that is, the samples are not i.i.d. Let $\Omega_{t-1}=\{X_{t-1}, A_{t-1}, Y_{t-1}, \dots, X_{1}, A_1, Y_{1}\}$ be the history with the space $\mathcal{M}_{t-1}$. The probability $p_t(a\mid x)$ is determined by a \emph{logging policy} $\pi_t:\mathcal{A}\times\mathcal{X}\times\mathcal{M}_{t-1}\to(0,1)$. We also assume that $\pi_t$ is conditionally independent of $Y_t(a)$ to satisfy the unconfoundedness (Remark~\ref{rem:unconfoundedness}).

\begin{remark}[Stable unit treatment value assumption]
The DGP also implies the \emph{stable unit treatment value assumption}, that is, $p(y(a)\mid x)$ is invariant for any $p_t(a\mid x)$ \citep{Rubi:86}.
\end{remark}

\begin{remark}[Unconfoundedness]
\label{rem:unconfoundedness}
In this paper, unconfoundedness refers independence between $(Y_t(1),\dots, Y_t(K))$ and $A_t$ conditioned on $X_t$ and $\Omega_{t-1}$, which is required for identification of the mean outcome. 
\end{remark}

\subsection{Parameter of Interest in Mean Outcome Estimation}
Let a function $\epol:\mathcal{A}\times \mathcal{X} \to\mathbb{R}$ be an \emph{evaluation weight}. We consider estimating the mean outcome weighted by an evaluation weight $\pi^{\mathrm{e}}(a \mid x)$ defined as 
\begin{align*}
R(\epol) := \mathbb{E}\left[\sum^K_{a=1}\epol(a \mid x)Y_t(a)\right].
\end{align*}
\citet{dudik2011doubly} regarded the weight as an policy that we want to evaluate by limiting the range into $(0,1)$ and the sum to $1$. The ATE is also a special case of the mean outcome for two actions $\mathcal{A}=\{1,2\}$, where $\epol(1\mid x) = 1$ and $\epol(2\mid x) = -1$. To identify $R(\epol)$, we assume the boundedness of the potential outcome. 

\begin{assumption}
\label{asm:outcome_boundedness}
For all $a\in\mathcal{A}$ and $t\in\{1,2,\dots,T\}$, there exists a constant $C_Y$ such that $|Y_t(a)| \leq C_Y$
\end{assumption}

\paragraph{Notations.} Let us denote $\mathbb{E}[Y_t(a)\mid x]$ and $\mathrm{Var}(Y_t(a)\mid x)$ as $f^*(a, x)$ and $v^*(a, x)$, respectively. Let $\hat{f}_{t}(a, x)$ be an estimators of $f^*(a, x)$ constructed from $\Omega_{t}$. Let $\mathcal{N}(\mu, \mathrm{var})$ be the normal distribution with the mean $\mu$ and the variance $\mathrm{var}$. For a random variable $Z$ and function $\mu$, let $\|\mu(Z)\|_2=\int |\mu(z)|^2 p(z) dz $ be the $L^{2}$-norm.

\section{Preliminaries of Mean Outcome Estimation}

\subsection{Mean Outcome Estimation from Independent Samples}
We introduce well-known estimators for a standard mean outcome estimation problem with i.i.d. samples where $\pi_1(a\mid x, \Omega_{0})=\pi_2(a\mid x, \Omega_{1})=\cdots=\pi(a\mid x)$. One of the standard estimators is an inverse probability weighting (IPW) estimator 
\begin{align*}
\widehat{R}^{\mathrm{IPW}}_T(\epol)=\frac{1}{T}\sum^T_{t=1}\sum^K_{a=1}\frac{\epol(a\mid X_t)\mathbbm{1}[A_t=a]Y_t}{\pi_{t}(a\mid X_t)},
\end{align*}
which are also called importance sampling \citep{Horvitz1952}. If $\hat{f}$ is a consistent estimator of $f^*$, the direct method (DM) estimator defined as $\sum^K_{a=1}\epol(a \mid x)\hat{f}(a\mid X_t)$ is known to be consistent to the policy value $R(\epol)$. By extending an IPW, \citet{robins94}, \citet{Scharfstein1999}, and \citet{Robins1999} proposed an Augmented IPW (AIPW) estimator defined as
\begin{align*}
\widehat{R}^{\mathrm{AIPW}}_T(\epol) = \frac{1}{T}\sum^T_{t=1}\sum^K_{a=1}\Bigg\{\frac{\epol(a\mid X_t)\mathbbm{1}[A_t=a]\left(Y_t - \hat{f}(a, X_t)\right) }{\pi(a\mid X_t)} + \epol(a\mid X_t)\hat{f}(a, X_t)\Bigg\},
\end{align*}
where $\hat{f}$ is a consistent estimator of $f^*$. In addition, a doubly robust (DR) estimator is also a standard choice \citep{Scharfstein1999,Bang2005}, which is defined as 
\begin{align*}
\widehat{R}^{\mathrm{DR}}_T(\epol) = \frac{1}{T}\sum^T_{t=1}\sum^K_{a=1}\Bigg\{\frac{\epol(a\mid X_t)\mathbbm{1}[A_t=a]\left(Y_t - \hat{f}(a, X_t)\right) }{\hat{g}(a\mid X_t)}
 + \epol(a\mid X_t)\hat{f}(a, X_t)\Bigg\},
 \end{align*}
 where $\hat{g}$ is a consistent estimator of $\pi$.

\paragraph{Semiparametric efficiency bound.} In many cases, we are interested in the asymptotic efficiency of the estimators. The lower bound of the asymptotic variance is defined for an estimator under some posited models of the DGP. If this posited model is a parametric model, then the lower bound is equal to the \Cramer-Rao lower bound. When this posited model is a non- or semiparametric model, the corresponding lower bound can still be defined \citep{bickel98}. For OPV estimation setting, \citet{narita2019counterfactual} shows that the semiparametric lower bound of the DGP under $p_1(a\mid x)=\cdots=p_T(a\mid x)=p(a\mid x)$ is 
\begin{align*}
\Psi(\epol) = \mathbb{E}\Bigg[\sum^{K}_{a=1}\frac{\big(\epol(a\mid X_t)\big)^2v^*(a, X_t)}{p(a\mid X_t)}+ \left(\sum^{K}_{a=1}\epol(a\mid X_t)f^*(a, X_t) - R(\epol)\right)^2\Bigg].
\end{align*}
The asymptotic variance of the asymptotic distribution is also known as the asymptotic mean squared error (MSE). By constructing a mean outcome estimator achieving the semiparametric lower bound, we can also minimize the MSE between the estimator and the true value $R(\epol)$, not only obtain a tight confidence interval.

\subsection{Mean Outcome Estimation from Dependent Samples}
There are mainly three approaches for deriving the asymptotic normality of a mean outcome estimator from dependent samples: (i) assuming the convergence of the logging policy $\pi_t(a\mid x, \Omega_{t-1})$ to a time-invariant probability \citep{Laan2014onlinetml,hadad2019,Kato2020adaptive}; (ii) assuming the presence of batched samples \citep{Hahn2011,Laan2014onlinetml,kelly2020batched}; (iii) standardizing the score functions. Under the first approach, \citet{Laan2008TheCA} and \citet{Kato2020adaptive} put the following assumption.
\begin{assumption}\label{asm:old_assumption}
For all $a\in \mathcal{A}$ and $x\in \mathcal{X}$, $\hat{f}_{t-1}(a, x) \xrightarrow{\mathrm{p}} f^*(a, x)$ and $\pi_t(a\mid x, \Omega_{t-1})\xrightarrow{\mathrm{p}} \alpha(a\mid x)$, where $\alpha: \mathcal{A}\times\mathcal{X}\to(0,1)$ is a time-invariant function such that $\sum^K_{a'=1}\alpha(a'\mid x)=1$ and there exists a constant $C_\pi$ satisfying $\left| \frac{\epol(a\mid x)}{\pi_t(a\mid x, \Omega_{t-1})}\right| \leq C_{\pi}$ for all $a\in \mathcal{A}$ and $x\in \mathcal{X}$, and $\Omega_{t-1}\in\mathcal{M}_{t-1}$.
\end{assumption}
Then, \citet{Laan2008TheCA} proposed the adaptive version of an IPW (AdaIPW) estimator defined as $R^{\mathrm{AdaIPW}}_T(\epol)=\frac{1}{T}\sum^T_{t=1}\frac{\epol(A_t\mid X_t)\mathbbm{1}[A_t=a]Y_t }{\pi_{t}(A_t\mid X_t, \Omega_{t-1})}$ and \citet{Laan2008TheCA} and \citet{Kato2020adaptive} proposed estimators based on an Adaptive AIPW (A2IPW) estimator $\widehat{R}^{\mathrm{A2IPW}}_T(\epol)$ defined as
\begin{align*}
\frac{1}{T}\sum^T_{t=1}\sum^K_{a=1}\Bigg\{\frac{\epol(a\mid X_t)\mathbbm{1}[A_t=a]\left(Y_t - \hat{f}_{t-1}(a, X_t)\right) }{\pi_{t}(a\mid X_t, \Omega_{t-1})}+ \epol(a\mid X_t)\hat{f}_{t-1}(a, X_t)\Bigg\},
\end{align*}
where $\hat{f}_t$ is a consistent estimator of $f^*$ constructed only using $\Omega_{t-1}$. Under Assumption~\ref{asm:old_assumption}.  \citet{Kato2020adaptive} showed the asymptotic normality of an A2IPW estimator.
\begin{proposition}[Asymptotic distribution of an A2IPW estimator]
\label{prp:asymp_dist_a2ipw}
Under Assumptions~\ref{asm:outcome_boundedness} and \ref{asm:old_assumption}, $\sqrt{T}\left(\widehat{R}^{\mathrm{A2IPW}}_T(\epol)-R(\epol)\right)\xrightarrow{d}\mathcal{N}\left(0, \Psi(\alpha)\right)$.
\end{proposition}
Besides, By replacing the true logging policy $\pi_t$ with its estimator $\hat{g}_{t-1}$, \citet{kato2020theoreticalcomparison} proposed an ADR estimator $\widehat{R}^{\mathrm{ADR}}_T(\epol)$ defined as 
\begin{align*}
\frac{1}{T}\sum^T_{t=1}\sum^K_{a=1}\Bigg\{\frac{\epol(a\mid X_t)\mathbbm{1}[A_t=a]\left(Y_t - \hat{f}_{t-1}(a, X_t)\right) }{\hat{g}_{t-1}(a\mid X_t)}  + \epol(a\mid X_t)\hat{f}_{t-1}(a, X_t)\Bigg\}.
\end{align*}
\citet{kato2020theoreticalcomparison} showed Proposition~\ref{prp:asymp_dist_adr} using sample-fitting, which is also used in \citet{Laan2014onlinetml}. First, we assume boundedness and consistencies of nuisance estimators $\hat{f}_{t-1}$ and $\hat{g}_{t-1}$.
\begin{assumption}
\label{asm:boundedness_est}
For all $a\in \mathcal{A}$ and $x\in\mathcal{X}$, there exist constants $C_f, C_g > 0$ such that $|\hat{f}_{t-1}(a, x)| \leq C_f$ and $\left|\frac{\epol(a\mid x)}{\hat{g}_{t-1}(a\mid x)}\right| \leq C_g$ for all $a\in\mathcal{A}$, $x\in\mathcal{X}$, and $t\in\{1,2,\dots,T\}$.
\end{assumption}
\begin{assumption}
\label{asm:consistency}
For all $a\in \mathcal{A}$, $\|\hat{g}_{t-1}(a\mid X_t) - \alpha(a\mid X_t)\|_{2}=\op(1)$, $\|\hat{f}_{t-1}(a,X_t)-f^*(a,X_t)\|_2=\op(1)$.
\end{assumption}
In addition, we put the following assumption on the convergence rate.
\begin{assumption}
\label{asm:conv_rate}
For all $a\in \mathcal{A}$, $\|\hat{g}_{t-1}(a\mid X_t) - \alpha(a\mid X_t)\|_{2}\|\hat{f}_{t-1}(a,X_t)-f^*(a,X_t)\|_2=\op(t^{-1/2})$.
\end{assumption}
Then, \citet{kato2020theoreticalcomparison} proved the following proposition.
\begin{proposition}[Asymptotic normality of an ADR estimator]
\label{prp:asymp_dist_adr}
Then under Assumptions~\ref{asm:outcome_boundedness}--\ref{asm:conv_rate}, for the ADR estimator, 
\begin{align*}
\sqrt{T}\left(\widehat{R}^{\mathrm{ADR}}_T(\epol)-R(\epol)\right)\xrightarrow{d}\mathcal{N}\left(0, \Psi(\alpha)\right).
\end{align*}
\end{proposition}

\paragraph{Sample splitting and Donsker's condition.}
When estimating $f^*$ and $\frac{1}{t}\sum^t_{s=1}\pi_s$, we only use $\Omega_{t-1}$. Owing to this construction, we can derive the asymptotic normality of the semiparametric estimator without Donsker's condition. This technique is a variant of sample-splitting \citep{klaassen1987,ZhengWenjing2011CTME,ChernozhukovVictor2018Dmlf}. See \citet{Laan2014onlinetml} and \citet{kato2020theoreticalcomparison} for more details. 

\subsection{Conditions for Asymptotic Normality}
We consider a class of OPV estimators $\hat{R}_T$ such that there exists a function $\phi:\mathcal{X}\times\mathcal{A}\times \mathbb{R}\to \mathbb{R}$ satisfying 
\begin{align*}
\sqrt{T}\left(\hat{R}_T(\epol)\right) = \frac{1}{\sqrt{T}}\sum^T\phi(X_t, A_t, Y_t) + \op(1).
\end{align*}
Such an estimator $\hat{R}_t$ and function $\phi$ are called asymptotically linear estimator and influence function, respectively. If samples are i.i.d., an asymptotically linear estimator has an asymptotic normality as $\sqrt{T}\left(\hat{R}_T(\epol)\right)\to\mathcal{N}(0, \mathbb{E}\left[\phi(X_t, A_t, Y_t)\phi(X_t, A_t, Y_t)\right])$. However, when samples are dependent, we need to carefully consider the condition for asymptotic normality of $\frac{1}{\sqrt{T}}\sum^T\phi(X_t, A_t, Y_t)$. A standard strategy is to apply martingale CLT to $\frac{1}{\sqrt{T}}\sum^T\phi(X_t, A_t, Y_t)$. For a martingale difference sequence (MDS), the martingale CLT is provided as follows.
\begin{proposition}
\label{prp:marclt}[CLT for a martingale difference sequence, \citet{GVK126800421}, Proposition~7.9, p.~194] Let $\{R_t\}^\infty_{t=1}$ be a scalar martingale difference sequence with $\overline{R}_T=\frac{1}{T}\sum^T_{t=1}R_t$. Suppose that 
\begin{description}
\item[(a)] $\mathbb{E}[R^2_t] = \sigma^2_t$, a positive value with $(1/T)\sum^T_{t=1}\sigma^2_t\to\sigma^2$, a positive value; 
\item[(b)] $\mathbb{E}[|R_t|^r] < \infty$ for some $r>2$;
\item[(c)] $(1/T)\sum^{T}_{t=1}R^2_t\xrightarrow{p}\sigma^2$. 
\end{description}
Then $\sqrt{T}\overline{R}_T\xrightarrow{d}\mathcal{N}(\bm{0}, \sigma^2)$.
\end{proposition}
Here, the martingale CLT requires a mean outcome estimator a asymptotically constant variance, and there are several directions to construct estimators satisfying the conditions. For instance, convergence of the logging policy is an instance. In this paper, because the assumption is too restrictive, we consider more practical assumptions. 

\subsection{Related Work}
Compared with studies on mean outcome estimation for i.i.d. samples \citep{Horvitz1952,HahnJinyong1998OtRo,hirano2003efficient,Bang2005,dudik2011doubly,narita2019counterfactual,Bibaut2019moreffficient}. there are fewer studies on mean outcome estimation for not i.i.d. samples. When the logging policy converges, \citet{Laan2008TheCA}, \citet{Laan2014onlinetml}, and \citet{Luedtke2016} proposed IPW, AIPW, and DR type estimators mainly for ATE estimation. \citet{Laan2008TheCA} proposed an A2IPW estimator. \citet{Laan2014onlinetml} only implied a possibility of an ADR estimator, and \citet{kato2020theoreticalcomparison} showed it. Without the convergence assumption, the asymptotic normality still can be derived based on batched samples \citep{Hahn2011,Laan2014onlinetml} and standardization \citep{Luedtke2016}. Because an A2IPW estimator is unstable, \citet{hadad2019} proposed a stabilization method for an A2IPW estimator, and \citet{kato2020theoreticalcomparison} empirically showed that an ADR estimator is more stable than an A2IPW estimator, which has the same asymptotic distribution.

A semiparametric estimator usually requires Donsker's condition for its $\sqrt{N}$-consistency, where $N$ is a sample size \citep{bickel98}. For semiparametric inference without Donsker's condition, sample-splitting is a typical approach \citep{klaassen1987,ZhengWenjing2011CTME,ChernozhukovVictor2018Dmlf}, which is also referred to as \emph{cross-fitting}. As a variant of the sample-splitting for time-series, \citet{Laan2014onlinetml} and \citet{kato2020theoreticalcomparison} proposed \emph{adaptive-fitting}. \citet{Kallus2019IntrinsicallyES} also proposed mixingale-based sample-splitting. 

Finally, we introduce existing studies in other related topics. Adaptive importance sampling is a sample selection framework for efficient Monte Carlo simulation, similar to adaptive experiments \citep{Kloek1978,NAYLOR1988103,evans1988,oh1992,Cappe2008,Portier2018ais}. In causal inference, the conditional mean outcome is also a standard target, and \citet{Laan2008TheCA} and \citet{kelly2020batched} proposed methods for estimating it. The method of \citet{kelly2020batched} is a variant of the generalized method of moments for martingales \citep{hayashi2000}, which is also applied in \citet{Laan2014onlinetml}. \citep{Li2010,Li2011} proposed off-policy policy evaluation of (adaptive) MAB algorithm using i.i.d. samples generated from a random policy. 

\section{ADR Estimator when Average logging policy Converges}
\label{sec:adr}
Let us consider a situation where $\frac{1}{t}\sum^t_{s=1}\pi_s(a\mid x, \Omega_{s-1})$ converges to $\alpha(a\mid x)$  for all $x\in\mathcal{X}$ and $\{\Omega_{s-1}\}^t_{s=1}$ in probability. For this case, we consider an adaptive DR (ADR) estimator $\widehat{R}^{\mathrm{ADR}}_T(\epol)$ defined as
\begin{align*}
\frac{1}{T}\sum^T_{t=1}\sum^K_{a=1}\Bigg\{\frac{\epol(a\mid X_t)\mathbbm{1}[A_t=a]\left(Y_t - \hat{f}_{t-1}(a, X_t)\right) }{\hat{g}_{t-1}(a\mid X_t)} + \epol(a\mid X_t)\hat{f}_{t-1}(a, X_t)\Bigg\},
\end{align*}
where $\hat{g}_{t-1}(a\mid x)$ is an estimator of $\frac{1}{t}\sum^t_{s=1}\pi_s(a\mid x, \Omega_{s-1})$ or $\alpha(a\mid x)$, which is constructed only from $\Omega_{t-1}$. For instance, when minimizing the risk with the logistic loss, we can show that the solution is given as $\frac{1}{t}\sum^t_{s=1}\pi_s$. Let us consider the following risk of binary classification problem:
\begin{align*}
\int \frac{1}{t}\sum^t_{s=1}\Big(p(a=1)\log(h(x)) p_s(x\mid a=1) + p(a=2)\log(1-h(x)) p_s(x\mid a=2)\Big) dx,
\end{align*} 
where $p_s(x\mid a=1)$ is the conditional density of $x$ at the period $t$. By taking the derivative and first order condition, the minimizer is given as $h^*(x) = \frac{\frac{1}{t}\sum^t_{s=1}p_s(x\mid a=1)}{p(x)}=\frac{1}{t}\pi_s(a\mid x, \Omega_{s-1})$. By the law of large numbers for martingales, the risk can be approximated by 
\begin{align*}
\frac{1}{t}\sum^t_{s=1}\Big(\mathbbm{1}[A_s = 1]\log(h(X_s)) + \mathbbm{1}[A_s = 2]\log(1-h(X_s))\Big).
\end{align*} 
Therefore, by naively applying the logistic regression for $\{(X_s, A_s)\}^t_{s=1}$, we can obtain the consistent estimator $\hat{g}$. This estimator

\subsection{Convergence of the Average logging policy}
In this paper, we show that even though the assumption does not hold, we can derive the asymptotic normality of a mean outcome estimator under the following alternative assumption that the average logging policy converges in probability. 

\begin{assumption}\label{asm:mean_stat}
For all $a\in \mathcal{A}$, as $ t \to \infty$,
\begin{align*}
&\Big\|\frac{1}{t}\sum^t_{s=1}\pi_s(a\mid X, \Omega_{s-1})- \bar{\alpha}(a\mid X)\Big\|_2 \Big\|f^*(a,X)-\hat{f}_{t-1}(a,X)\Big\|_2=\op(t^{-1/2}),\\
&\|\hat{g}_{t-1}(a\mid X_t) - \bar{\alpha}(a\mid X_t)\|_{2}\|f^*(a,X_t)-\hat{f}_{t-1}(a,X_t)\|_2=\op(t^{-1/2}),
\end{align*}
where $\bar{\alpha}: \mathcal{A}\times\mathcal{X}\to(0,1)$ is a time-invariant function such that $\sum^K_{a'=1}\bar{\alpha}(a'\mid x)=1$ for all $x\in\mathcal{X}$ and there exists a constant $C_\alpha$ satisfying $\left| \frac{\epol(a\mid x)}{\alpha(a\mid x)} \right| < C_\alpha$, and the expectation of the norm is defined over $X_t$.
\end{assumption}

\begin{assumption}\label{asm:stationarity}
For all $a\in \mathcal{A}$, as $ T \to \infty$,
\begin{align*}
\Bigg|\frac{1}{T}\sum^T_{t=1} \mathbb{E}\Bigg[\epol(a\mid X)\left(\frac{\pi_{t}(a\mid X, \Omega_{t-1})}{\frac{1}{t}\sum^t_{s=1}\pi_s(a\mid X, \Omega_{s-1})} - 1\right)\left(f^*(a, X) - \hat{f}_{t-1}(a, X)\right)  \mid \Omega_{t-1}\Bigg]\Bigg| = \op(T^{-1/2}).
\end{align*}
\end{assumption}

Assumption~\ref{asm:mean_stat} is weaker than Assumption~\ref{asm:old_assumption} because Assumption~\ref{asm:mean_stat} holds under Assumption~\ref{asm:old_assumption}. Note that under Assumption~\ref{asm:mean_stat}, the logging policy $\pi_t$ can be deficient; that is $\pi_t$ can be $0$ at a period $t$. 

\subsection{Asymptotic Normality of an ADR Estimator}
\label{sec:adap_doubly_robust}
\citet{kato2020theoreticalcomparison} derived the asymptotic normality of the ADR estimator under Assumption~\ref{asm:old_assumption}. In this section, we show the asymptotic normality under Assumption~\ref{asm:mean_stat}. 
\begin{theorem}[Asymptotic normality of an ADR estimator when average logging policy converges]
\label{thm:asymp_normal}
Under Assumptions~\ref{asm:outcome_boundedness}, \ref{asm:boundedness_est}--\ref{asm:consistency} and ~\ref{asm:mean_stat}--\ref{asm:stationarity}, for the ADR estimator, $\sqrt{T}\left(\widehat{R}^{\mathrm{ADR}}_T(\epol)-R(\epol)\right)\xrightarrow{d}\mathcal{N}\left(0, \Psi(\bar{\alpha})\right)$.
\end{theorem}
To show Theorem~\ref{thm:asymp_normal}, we decompose $\sqrt{T}\left(\widehat{R}^{\mathrm{ADR}}_T(\epol) - R(\epol)\right)$ as 
\begin{align*}
\sqrt{T}\Big(\widehat{R}^{\mathrm{ADR}}_T(\epol) - \ddot{R}_T(\epol) + \ddot{R}_T(\epol) - R(\epol)\Big),
\end{align*}
where $\ddot{R}_T(\epol)$ is defined as
\begin{align*}
\frac{1}{T}\sum^T_{t=1}\sum^K_{a=1}\Bigg\{\frac{\epol(a\mid X_t)\mathbbm{1}[A_t=a]\left(Y_t - f^*(a, X_t)\right) }{\bar{\alpha}(a\mid X_t)} + \epol(a\mid X_t)f^*(a, X_t)\Bigg\}.
\end{align*}
The remaining problems are to show that 
\begin{align}
\label{asymp:first}
\sqrt{T}\Big(\widehat{R}^{\mathrm{ADR}}_T(\epol) - \ddot{R}_T(\epol)\Big)=\mathrm{o}_p(1)
\end{align}
and 
\begin{align}
\label{asymp:second}
\sqrt{T}\Big(\ddot{R}_T(\epol) - R(\epol)\Big)\xrightarrow{\mathrm{d}}\mathcal{N}\left(0, \bar{\sigma}^2\right).
\end{align}
We separately show \eqref{asymp:first} and \eqref{asymp:second} in Lemma~\ref{LEM:1} and \ref{LEM:2}, respectively. First, we show Lemma~\ref{LEM:1}.
\begin{lemma}
\label{LEM:1}
Under Assumptions~\ref{asm:outcome_boundedness}, \ref{asm:boundedness_est}--\ref{asm:consistency} and ~\ref{asm:mean_stat}--\ref{asm:stationarity},
$\sqrt{T}\Big(\widehat{R}^{\mathrm{ADR}}_T(\epol) - \ddot{R}_T(\epol)\Big)=\mathrm{o}_p(1)$
\end{lemma}
Lemma~\ref{LEM:1} is proved by a technique based on sample-splitting, as \citet{Laan2014onlinetml} and \citet{kato2020theoreticalcomparison}. Here, we show the sketch of proof. The full proof is shown in Appendix~\ref{appdx:proof:target1} of the supplementary material.
\begin{proof}[Sketch of proof]
Let us define 
\begin{align*}
&\phi_1(X_t, A_t, Y_t; g, f)=\sum^K_{a=1}\frac{\epol(a\mid X_t)\mathbbm{1}[A_t=a]\left(Y_t - f(a, X_t)\right) }{g(a\mid X_t)},\ \ \ \phi_2(X_t; f)=\sum^K_{a=1}\epol(a\mid X_t)f(a, X_t).
\end{align*}
We decompose $\sqrt{T}\Big(\widehat{R}^{\mathrm{ADR}}_T(\epol) - \ddot{R}_T(\epol)\Big)$ as 
\begin{align*}
&\widehat{R}^{\mathrm{ADR}}_T(\epol) - \ddot{R}(\epol)=\frac{1}{T}\sum^T_{t=1}\Bigg\{\phi_1(X_t, A_t, Y_t; \hat{g}_{t-1}, \hat{f}_{t-1}) - \phi_1(X_t, A_t, Y_t; \bar{\alpha}, f^*)\\
&\ \ \ -\mathbb{E}\left[\phi_1(X_t, A_t, Y_t; \hat{g}_{t-1}, \hat{f}_{t-1}) - \phi_1(X_t, A_t, Y_t; \bar{\alpha}, f^*)\mid \Omega_{t-1}\right]\\
&\ \ \ + \phi_2(X_t; \hat{f}_{t-1}) - \phi_2(X_t; f^*) -\mathbb{E}\left[\phi_2(X_t; \hat{f}_{t-1}) - \phi_2(X_t; f^*)\mid \Omega_{t-1}\right]\Bigg\}\\
&\ \ \ + \frac{1}{T}\sum^T_{t=1}\mathbb{E}\left[\phi_1(X_t, A_t, Y_t; \hat{g}_{t-1}, \hat{f}_{t-1})\mid \Omega_{t-1}\right] + \frac{1}{T}\sum^T_{t=1}\mathbb{E}\left[\phi_2(X_t; \hat{f}_{t-1})\mid \Omega_{t-1}\right]\\
&\ \ \ - \frac{1}{T}\sum^T_{t=1}\mathbb{E}\left[\phi_1(X_t, A_t, Y_t; \bar{\alpha}, f^*)\mid \Omega_{t-1}\right] - \frac{1}{T}\sum^T_{t=1}\mathbb{E}\left[\phi_2(X_t; f^*)\mid \Omega_{t-1}\right].
\end{align*}
In the following parts, we separately show that
\begin{align}
\label{eq:part1_main}
&\sqrt{T}\frac{1}{T}\sum^T_{t=1}\Bigg\{\phi_1(X_t, A_t, Y_t; \hat{g}_{t-1}, \hat{f}_{t-1}) - \phi_1(X_t, A_t, Y_t; \bar{\alpha}, f^*)\\
&\ \ \ -\mathbb{E}\left[\phi_1(X_t, A_t, Y_t; \hat{g}_{t-1}, \hat{f}_{t-1}) - \phi_1(X_t, A_t, Y_t; \bar{\alpha}, f^*)\mid \Omega_{t-1}\right]\nonumber\\
&\ \ \ + \phi_2(X_t; \hat{f}_{t-1}) - \phi_2(X_t; f^*) -\mathbb{E}\left[\phi_2(X_t; \hat{f}_{t-1}) - \phi_2(X_t; f^*)\mid \Omega_{t-1}\right]\Bigg\}= \mathrm{o}_p(1);\nonumber
\end{align}
and 
\begin{align}
\label{eq:part2_main}
&\frac{1}{T}\sum^T_{t=1}\mathbb{E}\left[\phi_1(X_t, A_t, Y_t; \hat{g}_{t-1}, \hat{f}_{t-1})\mid \Omega_{t-1}\right] + \frac{1}{T}\sum^T_{t=1}\mathbb{E}\left[\phi_2(X_t; \hat{f}_{t-1})\mid \Omega_{t-1}\right]\\
&- \frac{1}{T}\sum^T_{t=1}\mathbb{E}\left[\phi_1(X_t, A_t, Y_t; \bar{\alpha}, f^*)\mid \Omega_{t-1}\right] - \frac{1}{T}\sum^T_{t=1}\mathbb{E}\left[\phi_2(X_t; f^*)\mid \Omega_{t-1}\right] = \mathrm{o}_p(1/\sqrt{T}).\nonumber
\end{align}
To show \eqref{eq:part1_main}, we show the mean of the LHS of \eqref{eq:part1_main} is $0$ and its variance converges to $0$ in probability. To show \eqref{eq:part2_main}, we use Assumption~\ref{asm:mean_stat}.
\end{proof}

Next, Lemma~\ref{LEM:2} provides the asymptotic normality of $\ddot{R}_T(\epol)$.
\begin{lemma}
\label{LEM:2}
Under Assumptions \ref{asm:outcome_boundedness} and \ref{asm:mean_stat},
$\sqrt{T}\Big(\ddot{R}_T(\epol) - R(\epol)\Big)\xrightarrow{\mathrm{d}}\mathcal{N}\left(0, \Psi(\bar{\alpha}\right)$
\end{lemma}
Here, we show the sketch of proof. The full proof is shown in Appendix~\ref{appdx:proof:target2} of the supplementary material.
\begin{proof}[Sketch of proof]
The proof procedure follows \citet{Kato2020adaptive}. Let $\Gamma_t(a)$ be
\begin{align*}
\Gamma_t(a; \epol) = \frac{\epol(a\mid X_t)\mathbbm{1}[A_t = a](Y_t(a) - f^*(a\mid X_t))}{\bar{\alpha}(a\mid X_t)} - \epol(a\mid X_t)f^*(a\mid X_t).
\end{align*}
Note that $\ddot{R}_T(\epol) = \frac{1}{T}\sum^T_{t=1}\sum^K_{a=1}\Gamma_t(a; \epol)$. Then, for $Z_t = \sum^K_{a=1}\Gamma_t(a; \epol) - R(\epol)$, we want to show that
\begin{align*}
\sqrt{T}\left(\ddot{R}_T(\epol) - R(\epol)\right) = \sqrt{T}\left(\frac{1}{T}\sum^T_{t=1}Z_t\right) \xrightarrow{\mathrm{d}} \mathcal{N}\left(0, \sigma^2\right).
\end{align*}

Then, the sequence $\{Z_t\}^T_{t=1}$ is an MDS; that is, 
\begin{align*}
&\mathbb{E}\big[Z_t\mid \Omega_{t-1}\big]= \mathbb{E}\left[\sum^K_{a=1}\Gamma_t(a; \epol) - R(\epol)\mid \Omega_{t-1}\right]
\\
&= \mathbb{E}\left[\sum^K_{a=1}\epol(a\mid X_t)f^*(a, X_t) - R(\epol_t)\mid \Omega_{t-1}\right]\\
&\ \ \  + \mathbb{E}\left[\sum^K_{a=1}\frac{\epol(a\mid X_t)\mathbbm{1}[A_t = a](Y_t(a) - f^*(a, X_t))}{\bar{\alpha}(A_t\mid X_t)}\mid \Omega_{t-1}\right]
\\
&= 0 + \mathbb{E}\left[\mathbb{E}\left[\sum^K_{a=1}\frac{\epol(a \mid X_t)\pi(a\mid X_t, \Omega_{t-1})(f^*(a, X_t) - f^*(a, X_t))}{\bar{\alpha}(a\mid X_t)}\mid X_t, \Omega_{t-1}\right]\mid \Omega_{t-1}\right] = 0.
\end{align*}

Therefore, to derive the asymptotic distribution, we consider applying the CLT for a MDS introduced in Proposition~\ref{prp:marclt}. There are  following three conditions in the statement.

\begin{description}
\item[(a)] $\mathbb{E}\big[Z^2_t\big] = \nu^2_t > 0$ with $\big(1/T\big) \sum^T_{t=1}\nu^2_t\to \nu^2 > 0$;
\item[(b)] $\mathbb{E}\big[|Z_t|^r\big] < \infty$ for some $r>2$;
\item[(c)] $\big(1/T\big)\sum^T_{t=1}Z^2_t\xrightarrow{\mathrm{p}} \nu^2$. 
\end{description}
Because we assumed the boundedness of $z_t$ by assuming the boundedness of $Y_t$, $f^*$, and $\epol/\bar{\alpha}$, the condition~(b) holds. Therefore, the remaining task is to show the conditions~(a) and (c) hold. Here, the convergence of the average logging policy has an important role by making th variance time-invariance asymptotically.
\end{proof}

Then, from Lemma~\ref{LEM:1} and \ref{LEM:2}, we can show Theorem~\ref{thm:asymp_normal}.

\subsection{Asymptotic Normality of a Modified ADR Estimator}
To show the asymptotic normality of an ADR estimator, we need to check Assumption~\ref{asm:stationarity}, but it is not easy in practice. In this section, we modify the ADR estimator to guarantee the asymptotic normality more easily. We define a Modified ADR (MADR) $\widetilde{R}^{\mathrm{MADR}}_T(\epol)$ as
\begin{align*}
\frac{1}{T}\sum^T_{t=1}\sum^K_{a=1}\Bigg\{\frac{\epol(a\mid X_t)\mathbbm{1}[A_t=a]\left(Y_t - \tilde{f}_{t-1, u(T)}(a, X_t)\right) }{\hat{g}_{t-1}(a\mid X_t)} + \epol(a\mid X_t)\tilde{f}_{t-1, u(T)}(a, X_t)\Bigg\},&
\end{align*}
where $u(T)$ is a function of $T$ and
\begin{align*}
\tilde{f}_{t-1, u(T)}(a, X_t) = \begin{cases}
\hat{f}_{t-1}(a, X_t) & \mathrm{if}\ t\leq u(T)\\
\hat{f}_{u(T)}(a, X_t) & \mathrm{otherwise}.
\end{cases}
\end{align*}
Then, we put the following assumption.
\begin{assumption}\label{asm:mean_stat2}
There exists a function $u(T) > 0$ such that for all $a\in \mathcal{A}$, $\frac{u(T)}{\sqrt{T}}\to 0$ as $T\to\infty$, and  for $T > u(T)$, 
\begin{align}
\label{eq:con1}
&\Bigg\|\frac{1}{T-u(T)-1}\sum^{T}_{t=u(T)+1}\frac{\pi_t(a\mid X, \Omega_{t-1})}{\hat{g}_{t-1}(a\mid X)} - 1\Bigg\|_2\Bigg\|f^*(a,X_t)-\hat{f}_{u(T)}(a,X)\Bigg\|_2=\op(T^{-1/2}),\\ 
\label{eq:con2}
&\frac{1}{T}\sum^T_{u(T)+1}\left\|\frac{1}{t}\sum^{t}_{s=1}\pi_s(a\mid X, \Omega_{s-1}) - \hat{g}_{t-1}(a\mid X_t)\right\|_2\left\|f^*(a, X) - \hat{f}_{u(T)}(a, X)\right\|_2 = \op(T^{-1/2}),
\end{align}
where $\bar{\alpha}: \mathcal{A}\times\mathcal{X}\to(0,1)$ is a time-invariant function such that $\sum^K_{a'=1}\bar{\alpha}(a'\mid x)=1$ and there exists a constant $C_\alpha$ satisfying $\left| \frac{\epol(a\mid x)}{\alpha(a\mid x)}\right| < C_\alpha$, and the expectation of the norm is over $X_t$.
\end{assumption}
Then, we show Lemma~\ref{LEM:3} on the asymptotic bias of the MADR estimator.
\begin{lemma}
\label{LEM:3}
Under Assumptions~\ref{asm:outcome_boundedness}, \ref{asm:boundedness_est}--\ref{asm:consistency} and \ref{asm:mean_stat2}, 
$\sqrt{T}\Big(\widetilde{R}^{\mathrm{ADR}}_T(\epol) - \ddot{R}_T(\epol)\Big)=\mathrm{o}_p(1)$
\end{lemma}
\begin{proof}[Sketch of proof]
By using $\phi_1(X_t, A_t, Y_t; g, f)$ and $\phi_2(X_t; f)=\sum^K_{a=1}\epol(a\mid X_t)f(a, X_t)$ defined in the proof of Lemma~\ref{LEM:1}, we decompose $\sqrt{T}\Big(\widehat{R}^{\mathrm{ADR}}_T(\epol) - \ddot{R}_T(\epol)\Big)$ as 
\begin{align*}
&\widehat{R}^{\mathrm{ADR}}_T(\epol) - \ddot{R}(\epol)=\frac{1}{T}\sum^T_{t=1}\Bigg\{\phi_1(X_t, A_t, Y_t; \hat{g}_{t-1}, \tilde{f}_{t-1,u(T)}) - \phi_1(X_t, A_t, Y_t; \bar{\alpha}, f^*)\\
&\ \ \ -\mathbb{E}\left[\phi_1(X_t, A_t, Y_t; \hat{g}_{t-1}, \tilde{f}_{t-1,u(T)}) - \phi_1(X_t, A_t, Y_t; \bar{\alpha}, f^*)\mid \Omega_{t-1}\right]\\
&\ \ \ + \phi_2(X_t; \tilde{f}_{t-1,u(T)}) - \phi_2(X_t; f^*)\\
&\ \ \  -\mathbb{E}\left[\phi_2(X_t; \tilde{f}_{t-1,u(T)}) - \phi_2(X_t; f^*)\mid \Omega_{t-1}\right]\Bigg\}\\
&\ \ \ + \frac{1}{T}\sum^T_{t=1}\mathbb{E}\left[\phi_1(X_t, A_t, Y_t; \hat{g}_{t-1}, \tilde{f}_{t-1,u(T)})\mid \Omega_{t-1}\right]\\
&\ \ \  + \frac{1}{T}\sum^T_{t=1}\mathbb{E}\left[\phi_2(X_t; \tilde{f}_{t-1,u(T)})\mid \Omega_{t-1}\right]\\
&\ \ \ - \frac{1}{T}\sum^T_{t=1}\mathbb{E}\left[\phi_1(X_t, A_t, Y_t; \bar{\alpha}, f^*)\mid \Omega_{t-1}\right]\\
&\ \ \  - \frac{1}{T}\sum^T_{t=1}\mathbb{E}\left[\phi_2(X_t; f^*)\mid \Omega_{t-1}\right].
\end{align*}
Following the almost same process as the proof of Lemma~\ref{LEM:1}, we can show that
\begin{align*}
&\sqrt{T}\frac{1}{T}\sum^T_{t=1}\Bigg\{\phi_1(X_t, A_t, Y_t; \hat{g}_{t-1}, \tilde{f}_{t-1,u(T)}) - \phi_1(X_t, A_t, Y_t; \bar{\alpha}, f^*)\\
&\ \ \ -\mathbb{E}\left[\phi_1(X_t, A_t, Y_t; \hat{g}_{t-1}, \tilde{f}_{t-1,u(T)}) - \phi_1(X_t, A_t, Y_t; \bar{\alpha}, f^*)\mid \Omega_{t-1}\right]\nonumber\\
&\ \ \ + \phi_2(X_t; \tilde{f}_{t-1,u(T)}) - \phi_2(X_t; f^*)\\
&\ \ \  -\mathbb{E}\left[\phi_2(X_t; \tilde{f}_{t-1,u(T)}) - \phi_2(X_t; f^*)\mid \Omega_{t-1}\right]\Bigg\}= \mathrm{o}_p(1);\nonumber
\end{align*}
Therefore, we consider showing
\begin{align}
\label{sketch:proof:1}
&\frac{1}{T}\sum^T_{t=1}\mathbb{E}\left[\phi_1(X_t, A_t, Y_t; \hat{g}_{t-1}, \tilde{f}_{t-1,u(T)})\mid \Omega_{t-1}\right]\\
&\ \ \  + \frac{1}{T}\sum^T_{t=1}\mathbb{E}\left[\phi_2(X_t; \tilde{f}_{t-1,u(T)})\mid \Omega_{t-1}\right]\nonumber\\
&\ \ \ - \frac{1}{T}\sum^T_{t=1}\mathbb{E}\left[\phi_1(X_t, A_t, Y_t; \bar{\alpha}, f^*)\mid \Omega_{t-1}\right]\\
&\ \ \  - \frac{1}{T}\sum^T_{t=1}\mathbb{E}\left[\phi_2(X_t; f^*)\mid \Omega_{t-1}\right] = \mathrm{o}_p(1/\sqrt{T}).
\end{align}
If \eqref{sketch:proof:1} holds, then we can prove the statement.

First, as shown in the proof of Lemma~\ref{LEM:1}, we bound the LHS of \eqref{sketch:proof:1} as follows:
\begin{align}
&\frac{1}{T}\sum^T_{t=1}\mathbb{E}\left[\phi_1(X_t, A_t, Y_t; \hat{g}_{t-1}, \tilde{f}_{t-1,u(T)})\mid \Omega_{t-1}\right] + \frac{1}{T}\sum^T_{t=1}\mathbb{E}\left[\phi_2(X_t; \tilde{f}_{t-1,u(T)})\mid \Omega_{t-1}\right]\nonumber\\
&\ \ \ - \frac{1}{T}\sum^T_{t=1}\mathbb{E}\left[\phi_1(X_t, A_t, Y_t; \bar{\alpha}, f^*)\mid \Omega_{t-1}\right] - \frac{1}{T}\sum^T_{t=1}\mathbb{E}\left[\phi_2(X_t; f^*)\mid \Omega_{t-1}\right] = \mathrm{o}_p(1/\sqrt{T})\nonumber\\
&\leq \Bigg|\frac{1}{T}\sum^T_{t=1} \mathbb{E}\Bigg[\frac{\epol(a\mid X)\Big(\pi_{t}(a\mid X, \Omega_{t-1}) - \hat{g}_{t-1}(a\mid X)\Big)\left(f^*(a, X) - \tilde{f}_{t-1,u(T)}(a, X)\right) }{\hat{g}_{t-1}(a\mid X)}  \mid \Omega_{t-1}\Bigg]\Bigg|\nonumber\\
&\leq \Bigg|\frac{1}{T}\sum^T_{t=1} \mathbb{E}\Bigg[\frac{\epol(a\mid X)\Big(\frac{1}{t}\sum^t_{s=1}\pi_{s}(a\mid X, \Omega_{s-1}) - \hat{g}_{t-1}(a\mid X)\Big)\left(f^*(a, X) - \tilde{f}_{t-1,u(T)}(a, X)\right) }{\hat{g}_{t-1}(a\mid X)}  \mid \Omega_{t-1}\Bigg]\Bigg|\nonumber\\
&\  + \Bigg|\frac{1}{T}\sum^T_{t=1} \mathbb{E}\Bigg[\frac{\epol(a\mid X)\Big(\pi_{t}(a\mid X, \Omega_{t-1}) - \frac{1}{t}\sum^t_{s=1}\pi_{s}(a\mid X, \Omega_{s-1})\Big)\left(f^*(a, X) - \tilde{f}_{t-1,u(T)}(a, X)\right) }{\hat{g}_{t-1}(a\mid X)}  \mid \Omega_{t-1}\Bigg]\Bigg|\nonumber\\
\label{sketch:proof:2}
&\leq \frac{C}{T}\sum^T_{t=1} \Bigg|\mathbb{E}\Bigg[ \left(\frac{1}{t}\sum^t_{s=1}\pi_{s}(a\mid X, \Omega_{s-1}) - \hat{g}_{t-1}(a\mid X)\right)\left(f^*(a, X) - \tilde{f}_{t-1,u(T)}(a, X)\right)\mid \Omega_{t-1}\Bigg]\Bigg|\\
&\ \ \ + \Bigg|\frac{1}{T}\sum^T_{t=1} \mathbb{E}\Bigg[\epol(a\mid X)\left(\frac{\pi_{t}(a\mid X, \Omega_{t-1})}{\hat{g}_{t-1}(a\mid X)} - \frac{\frac{1}{t}\sum^t_{s=1}\pi_{s}(a\mid X, \Omega_{s-1})}{\hat{g}_{t-1}(a\mid X)}\right)\nonumber\\
\label{sketch:proof:3}
&\ \ \ \ \ \ \ \ \ \ \ \ \ \ \ \ \ \ \ \ \ \ \ \ \ \ \ \ \ \ \ \ \ \ \ \ \ \ \ \ \ \ \ \ \ \ \ \ \ \ \ \ \ \ \ \ \ \ \ \ \ \ \ \ \ \ \ \ \ \ \ \ \times\left(f^*(a, X) - \tilde{f}_{t-1,u(T)}(a, X)\right) \mid \Omega_{t-1}\Bigg]\Bigg|,
\end{align}
where $C > 0$ is a constant. We can show that the first term \eqref{sketch:proof:2} converges with $\op(T^{-1/2})$. By using the \Holder's inequality $\|\mu\nu \|_1 \leq  \|\mu \|_2  \|\nu \|_2$,
\begin{align*}
& \frac{1}{T}\sum^T_{t=1} \Bigg|\mathbb{E}\Bigg[ \left(\frac{1}{t}\sum^t_{s=1}\pi_{s}(a\mid X, \Omega_{s-1}) - \hat{g}_{t-1}(a\mid X)\right)\left(f^*(a, X) - \tilde{f}_{t-1,u(T)}(a, X)\right)\mid \Omega_{t-1}\Bigg]\Bigg|\nonumber\\
& \leq \frac{1}{T}\sum^T_{t=1}\left\|\frac{1}{t}\sum^t_{s=1}\pi_{s}(a\mid X, \Omega_{s-1}) - \hat{g}_{t-1}(a\mid X)\right\|_2\left\|f^*(a, X) - \tilde{f}_{t-1,u(T)}(a, X)\right\|_2\nonumber\\
&\leq \frac{1}{T}\sum^T_{t=1}\left(\left\|\frac{1}{t}\sum^t_{s=1}\pi_{s}(a\mid X, \Omega_{s-1}) - \bar{\alpha}(a\mid X)\right\|_2 + \left\|\bar{\alpha}(a\mid X) - \hat{g}_{t-1}(a\mid X)\right\|_2\right)\nonumber\\
&\ \ \ \ \ \ \ \ \ \ \ \ \ \ \ \ \ \ \ \ \ \ \ \ \ \ \ \ \ \ \ \ \ \ \ \ \ \ \ \ \ \ \ \ \ \ \times \left\|f^*(a, X) - \tilde{f}_{t-1,u(T)}(a, X)\right\|_2\\
&= \frac{1}{T}\sum^T_{t=1}\Bigg(\left\|\frac{1}{t}\sum^t_{s=1}\pi_{s}(a\mid X, \Omega_{s-1}) - \bar{\alpha}(a\mid X)\right\|_2\left\|f^*(a, X) - \tilde{f}_{t-1,u(T)}(a, X)\right\|_2\nonumber\\
&\ \ \ \ \ \ \ \ \ \ \ \ \ \ \ \ \ \ \ \ \ \ \ \ \ \ \ \ \ \ \ \ \ \ \ \ \ \ \ \ \ \ \ \ \ \ + \left\|\bar{\alpha}(a\mid X) - \hat{g}_{t-1}(a\mid X)\right\|_2\left\|f^*(a, X) - \tilde{f}_{t-1,u(T)}(a, X)\right\|_2\Bigg).
\end{align*}
From the definition of $u(T)$, we can bound it as
\begin{align*}
&\frac{1}{T}\sum^{u(T)}_{t=1}\Bigg(\left\|\frac{1}{t}\sum^t_{s=1}\pi_{s}(a\mid X, \Omega_{s-1}) - \bar{\alpha}(a\mid X)\right\|_2\left\|f^*(a, X) - \hat{f}_{t-1}(a, X)\right\|_2\\
&\ \ \ \ \ \ \ \ \ \ \ \ \ \ \ \ \ \ \ \ \ \ \ \ \ \ \ \ \ \ \ \ \ \ \ \  + \left\|\bar{\alpha}(a\mid X) - \hat{g}_{t-1}(a\mid X)\right\|_2\left\|f^*(a, X) - \hat{f}_{t-1}(a, X)\right\|_2\Bigg)\\
&\ \ \ + \frac{1}{T}\sum^T_{t=u(T)+1}\Bigg(\left\|\frac{1}{t}\sum^t_{s=1}\pi_{s}(a\mid X, \Omega_{s-1}) - \bar{\alpha}(a\mid X)\right\|_2\left\|f^*(a, X) - \hat{f}_{u}(a, X)\right\|_2\\
&\ \ \ \ \ \ \ \ \ \ \ \ \ \ \ \ \ \ \ \ \ \ \ \ \ \ \ \ \ \ \ \ \ \ \ \  + \left\|\bar{\alpha}(a\mid X) - \hat{g}_{t-1}(a\mid X)\right\|_2\left\|f^*(a, X) - \hat{f}_{u}(a, X)\right\|_2\Bigg)\\
& \leq C \frac{u(T)}{T} + \frac{1}{T}\sum^{T}_{t=u(T)+1}\op(T^{-1/2}) +\frac{1}{T}\sum^{T}_{t=u(T)+1}\op(T^{-1/2})\\
& = \op(T^{-1/2}) + \op(T^{-1/2}) = \op(T^{-1/2}),
\end{align*}
where $C > 0$ is a constant. 

Next, we show that the second term \eqref{sketch:proof:3} converges with $\op(T^{-1/2})$. By using $u(T)$ of the statement, we have
\begin{align}
&\Bigg|\frac{1}{T}\sum^T_{t=1} \mathbb{E}\Bigg[\epol(a\mid X)\left(\frac{\pi_{t}(a\mid X, \Omega_{t-1})}{\hat{g}_{t-1}(a\mid X)} - \frac{\frac{1}{t}\sum^t_{s=1}\pi_{s}(a\mid X, \Omega_{s-1})}{\hat{g}_{t-1}(a\mid X)}\right)\nonumber\\
&\ \ \ \ \ \ \ \ \ \ \ \ \ \ \ \ \ \ \ \ \ \ \ \ \ \ \ \ \ \ \ \ \ \ \ \ \ \ \ \ \ \ \ \ \ \ \ \ \ \ \ \ \ \ \ \ \ \times\left(f^*(a, X) - \tilde{f}_{t-1,u(T)}(a, X)\right) \mid \Omega_{t-1}\Bigg]\Bigg|\nonumber\\
&\leq \Bigg|\frac{1}{T}\sum^{u(T)}_{t=1} \mathbb{E}\Bigg[\epol(a\mid X)\left(\frac{\pi_{t}(a\mid X, \Omega_{t-1})}{\hat{g}_{t-1}(a\mid X)} - \frac{\frac{1}{t}\sum^t_{s=1}\pi_{s}(a\mid X, \Omega_{s-1})}{\hat{g}_{t-1}(a\mid X)}\right)\nonumber\\
\label{sketch:proof:4}
&\ \ \ \ \ \ \ \ \ \ \ \ \ \ \ \ \ \ \ \ \ \ \ \ \ \ \ \ \ \ \ \ \ \ \ \ \ \ \ \ \ \ \ \ \ \ \ \ \ \ \ \ \ \ \ \ \ \times \left(f^*(a, X) - \tilde{f}_{t-1,u(T)}(a, X)\right) \mid \Omega_{t-1}\Bigg]\Bigg|\\
&\ \ \ + \Bigg|\frac{1}{T}\sum^T_{t=u(T)+1} \mathbb{E}\Bigg[\epol(a\mid X)\left(\frac{\pi_{t}(a\mid X, \Omega_{t-1})}{\hat{g}_{t-1}(a\mid X)} - \frac{\frac{1}{t}\sum^t_{s=1}\pi_{s}(a\mid X, \Omega_{s-1})}{\hat{g}_{t-1}(a\mid X)}\right)\nonumber\\
\label{sketch:proof:5}
&\ \ \ \ \ \ \ \ \ \ \ \ \ \ \ \ \ \ \ \ \ \ \ \ \ \ \ \ \ \ \ \ \ \ \ \ \ \ \ \ \ \ \ \ \ \ \ \ \ \ \ \ \ \ \ \ \ \times \left(f^*(a, X) - \tilde{f}_{t-1,u(T)}(a, X)\right) \mid \Omega_{t-1}\Bigg]\Bigg|.
\end{align}
Because all variables are bounded, for a constant $C > 0$, the first term \eqref{sketch:proof:4} is bounded as
\begin{align*}
&\Bigg|\frac{1}{T}\sum^{u(T)}_{t=1} \mathbb{E}\Bigg[\epol(a\mid X)\left(\frac{\pi_{t}(a\mid X, \Omega_{t-1})}{\hat{g}_{t-1}(a\mid X)} - \frac{\frac{1}{t}\sum^t_{s=1}\pi_{s}(a\mid X, \Omega_{s-1})}{\hat{g}_{t-1}(a\mid X)}\right)\\
&\ \ \ \ \ \ \ \ \ \ \ \ \ \ \ \ \ \ \ \ \ \ \ \ \ \ \ \ \ \ \ \ \ \ \ \ \ \ \ \ \ \times\left(f^*(a, X) - \tilde{f}_{t-1,u(T)}(a, X)\right) \mid \Omega_{t-1}\Bigg]\Bigg|\\
&\leq C\frac{u(T)}{T}.
\end{align*}
Here, from the definition of $u(T)$, $\frac{u(T)}{\sqrt{T}}\to 0$ as $T\to \infty$. Then, we consider bounding the second term \eqref{sketch:proof:5}. First, we bound it as
\begin{align}
&\Bigg|\frac{1}{T}\sum^T_{t=u(T)+1} \mathbb{E}\Bigg[\epol(a\mid X)\left(\frac{\pi_{t}(a\mid X, \Omega_{t-1})}{\hat{g}_{t-1}(a\mid X)} - \frac{\frac{1}{t}\sum^t_{s=1}\pi_{s}(a\mid X, \Omega_{s-1})}{\hat{g}_{t-1}(a\mid X)}\right)\nonumber\\
&\ \ \ \ \ \ \ \ \ \ \ \ \ \ \ \ \ \ \ \ \ \ \ \ \ \ \ \ \ \ \ \ \ \ \ \ \ \ \ \ \ \ \ \ \ \ \ \ \ \ \ \ \ \times\left(f^*(a, X) - \tilde{f}_{t-1,u(T)}(a, X)\right) \mid \Omega_{t-1}\Bigg]\Bigg|\nonumber\\
&=\Bigg|\frac{1}{T}\sum^T_{t=u(T)+1} \mathbb{E}\Bigg[\epol(a\mid X)\left(\frac{\pi_{t}(a\mid X, \Omega_{t-1})}{\hat{g}_{t-1}(a\mid X)} - \frac{\frac{1}{t}\sum^t_{s=1}\pi_{s}(a\mid X, \Omega_{s-1})}{\hat{g}_{t-1}(a\mid X)}\right)\nonumber\\
&\ \ \ \ \ \ \ \ \ \ \ \ \ \ \ \ \ \ \ \ \ \ \ \ \ \ \ \ \ \ \ \ \ \ \ \ \ \ \ \ \ \ \ \ \ \ \ \ \ \ \ \ \ \times\left(f^*(a, X) - \tilde{f}_{u(T)}(a, X)\right) \mid \Omega_{t-1}\Bigg]\Bigg|\nonumber\\
&=\Bigg|\frac{1}{T}\sum^T_{t=u(T)+1} \mathbb{E}\Bigg[\epol(a\mid X)\left(\frac{\pi_{t}(a\mid X, \Omega_{t-1})}{\hat{g}_{t-1}(a\mid X)} - 1\right)\left(f^*(a, X) - \tilde{f}_{u(T)}(a, X)\right) \mid \Omega_{t-1}\Bigg]\Bigg|\nonumber\\
&\ \ \ + \Bigg|\frac{1}{T}\sum^T_{t=u(T)+1} \mathbb{E}\Bigg[\epol(a\mid X)\left(1 - \frac{\frac{1}{t}\sum^t_{s=1}\pi_{s}(a\mid X, \Omega_{s-1})}{\hat{g}_{t-1}(a\mid X)}\right)\nonumber\\
&\ \ \ \ \ \ \ \ \ \ \ \ \ \ \ \ \ \ \ \ \ \ \ \ \ \ \ \ \ \ \ \ \ \ \ \ \ \ \ \ \ \ \ \ \ \ \ \ \ \ \ \ \ \times\left(f^*(a, X) - \tilde{f}_{u(T)}(a, X)\right) \mid \Omega_{t-1}\Bigg]\Bigg|\nonumber\\
\label{eq:bound:term1}
&\leq \Bigg|\frac{1}{T}\sum^T_{t=u(T)+1} \mathbb{E}\Bigg[\epol(a\mid X)\left(\frac{\pi_{t}(a\mid X, \Omega_{t-1})}{\hat{g}_{t-1}(a\mid X)} - 1\right)\nonumber\\
&\ \ \ \ \ \ \ \ \ \ \ \ \ \ \ \ \ \ \ \ \ \ \ \ \ \ \ \ \ \ \ \ \ \ \ \ \ \ \ \ \ \ \ \ \ \ \ \ \ \ \ \ \ \times\left(f^*(a, X) - \tilde{f}_{u(T)}(a, X)\right) \mid \Omega_{t-1}\Bigg]\Bigg|\\
\label{eq:bound:term2}
&\ \ \ + \frac{1}{T}\sum^T_{t=u(T)+1}  \Bigg|\mathbb{E}\Bigg[\epol(a\mid X)\left(1 - \frac{\frac{1}{t}\sum^t_{s=1}\pi_{s}(a\mid X, \Omega_{s-1})}{\hat{g}_{t-1}(a\mid X)}\right)\nonumber\\
&\ \ \ \ \ \ \ \ \ \ \ \ \ \ \ \ \ \ \ \ \ \ \ \ \ \ \ \ \ \ \ \ \ \ \ \ \ \ \ \ \ \ \ \ \ \ \ \ \ \ \ \ \ \times\left(f^*(a, X) - \tilde{f}_{u(T)}(a, X)\right) \mid \Omega_{t-1}\Bigg]\Bigg|.
\end{align}

We separately bound \eqref{eq:bound:term1} and \eqref{eq:bound:term2}. First, we bound \eqref{eq:bound:term1}. By using the \Holder's inequality $\|\mu\nu\|_1 \leq  \|\mu \|_2  \|\nu \|_2$,
\begin{align*}
&\Bigg|\frac{1}{T}\sum^T_{t=u(T)+1} \mathbb{E}\Bigg[\epol(a\mid X)\left(\frac{\pi_{t}(a\mid X, \Omega_{t-1})}{\hat{g}_{t-1}(a\mid X)} - 1\right)\left(f^*(a, X) - \tilde{f}_{u(T)}(a, X)\right) \mid \Omega_{t-1}\Bigg]\Bigg|\\
&=\Bigg|\frac{1}{T}\sum^T_{t=u(T)+1} \mathbb{E}_X\Bigg[\epol(a\mid X)\left(\frac{\pi_{t}(a\mid X, \Omega_{t-1})}{\hat{g}_{t-1}(a\mid X)} - 1\right)\left(f^*(a, X) - \tilde{f}_{u(T)}(a, X)\right)\Bigg]\Bigg|\\
&=\Bigg|\mathbb{E}_X\Bigg[\epol(a\mid X)\left(f^*(a, X) - \tilde{f}_{u(T)}(a, X)\right)\frac{1}{T}\sum^T_{t=u(T)+1}\left(\frac{\pi_{t}(a\mid X, \Omega_{t-1})}{\hat{g}_{t-1}(a\mid X)} - 1\right)\Bigg]\Bigg|\\
&\leq \frac{T-u(T)-1}{T}\Bigg\|\epol(a\mid X)\left(f^*(a, X) - \tilde{f}_{u(T)}(a, X)\right)\Bigg\|_2\\
&\ \ \ \ \ \ \ \ \ \ \ \ \ \ \ \ \ \ \ \ \ \ \ \ \ \ \ \ \ \ \ \ \ \ \ \ \ \ \ \ \ \times \Bigg\|\frac{1}{T-u(T)-1}\sum^T_{t=u(T)+1}\left(\frac{\pi_{t}(a\mid X, \Omega_{t-1})}{\hat{g}_{t-1}(a\mid X)} - 1\right)\Bigg\|_2\\
&\leq C\Bigg\|f^*(a, X) - \tilde{f}_{u(T)}(a, X)\Bigg\|_2\Bigg\|\frac{1}{T-u(T)-1}\sum^T_{t=u(T)+1}\frac{\pi_{t}(a\mid X, \Omega_{t-1})}{\hat{g}_{t-1}(a\mid X)} - 1\Bigg\|_2,
\end{align*}
where $\mathbb{E}_X$ denotes the expectation over $X$ and $C > 0$ is a constant. Then, we have $\eqref{eq:bound:term1}=\op(T^{-1/2})$ from Assumption~\ref{asm:mean_stat2}.

Second, by using Assumption~\ref{asm:mean_stat2}, we show that \eqref{eq:bound:term2} is $\op(T^{-1/2})$ as
\begin{align*}
&\frac{1}{T}\sum^T_{t=u(T)+1}  \Bigg|\mathbb{E}\Bigg[\epol(a\mid X)\left(1 - \frac{\frac{1}{t}\sum^t_{s=1}\pi_{s}(a\mid X, \Omega_{s-1})}{\hat{g}_{t-1}(a\mid X)}\right)\left(f^*(a, X) - \tilde{f}_{u(T)}(a, X)\right) \mid \Omega_{t-1}\Bigg]\Bigg|\\
&\leq \frac{C}{T}\sum^T_{t=u(T)+1}  \Bigg|\mathbb{E}\Bigg[\left(\hat{g}_{t-1}(a\mid X) - \frac{1}{t}\sum^t_{s=1}\pi_{s}(a\mid X, \Omega_{s-1})\right)\left(f^*(a, X) - \tilde{f}_{u(T)}(a, X)\right) \mid \Omega_{t-1}\Bigg]\Bigg|\\
&\leq \frac{C}{T}\sum^T_{t=u(T)+1}  \Bigg\|\hat{g}_{t-1}(a\mid X) - \frac{1}{t}\sum^t_{s=1}\pi_{s}(a\mid X, \Omega_{s-1})\Bigg\|_2\Bigg\|f^*(a, X) - \tilde{f}_{u(T)}(a, X)\Bigg\|_2.
\end{align*}
Then, the term is $\op(T^{-1/2})$ from Assumption~\ref{asm:mean_stat2}.
\end{proof}

By using Lemma~\ref{LEM:2} and \ref{LEM:3}, we can show the following theorem.
\begin{theorem}[Asymptotic normality of $\widetilde{R}^{\mathrm{MADR}}_T(\epol)$]
\label{thm:asymp_normal2}
For $u(t) > 0$ such that $\frac{u(T)}{\sqrt{T}}\to 0$, under Assumptions~\ref{asm:outcome_boundedness}, \ref{asm:boundedness_est}--\ref{asm:consistency} and \ref{asm:mean_stat2}, MADR estimator has the asymptotic normality as 
\begin{align*}
\sqrt{T}\left(\widetilde{R}^{\mathrm{MADR}}_T(\epol)-R(\epol)\right)\xrightarrow{d}\mathcal{N}\left(0, \Psi(\bar{\alpha}\right).
\end{align*}
\end{theorem}

Unlike an ADR estimator, a MADR estimator does not require Assumption~\ref{asm:stationarity}. This property is an advantage from the theoretical viewpoint. However, as shown in experiments, an ADR estimator shows better empirical performance. The remaining problem is to check that standard estimators satisfy Assumption~\ref{asm:mean_stat2}. Here, we show an example where nuisance estimators satisfy the requirement. For instance, we consider $u(T) = T^{1/3}$, which satisfies $\frac{u(T)}{\sqrt{T}}\to 0$ as $T\to\infty$. Under some conditions, sample averages and regression estimators have $\Op(T^{-1/2})$ convergence rate. Therefore, we can assume that there exist  $p, q, r<1/2$ such that 
\begin{align}
\label{asm:term:1}
&\Bigg\|\frac{1}{T-u(T)-1}\sum^{T}_{t=u(T)+1}\frac{\pi_t(a\mid X, \Omega_{t-1})}{\hat{g}_{t-1}(a\mid X)} - 1\Bigg\|_2 =\op(T^{-p}),\\
&\|f^*(a,X_t)-\hat{f}_{t-1}(a,X_t)\|_2=\op(t^{-q}),\ \ \ \left\|\bar{\alpha}(a\mid X) - \hat{g}_{t-1}(a\mid X)\right\|_2=\op(t^{-r}).\nonumber
\end{align} 
Here, note that $\Bigg\|\frac{1}{T-u(T)-1}\sum^{T}_{t=u(T)+1}\frac{\pi_t(a\mid X, \Omega_{t-1})}{\hat{g}_{t-1}(a\mid X)} - 1\Bigg\|_2$ is bounded by 
\begin{align*}
&\Bigg\|\frac{1}{\bar{\alpha}(a\mid X)}\left(\frac{1}{T-u(T)-1}\sum^{T}_{t=u(T)+1}\pi_t(a\mid X, \Omega_{t-1}) - \bar{\alpha}(a\mid X)\right)\Bigg\|_2\\
&\ \ \ + \Bigg\|\frac{1}{T-u(T)-1}\sum^{T}_{t=u(T)+1}\frac{\pi_t(a\mid X, \Omega_{t-1})\big\{\bar{\alpha}(a\mid X) - \hat{g}_{t-1}(a\mid X)\big\}}{\hat{g}_{t-1}(a\mid X)\bar{\alpha}(a\mid X)}\Bigg\|_2.
\end{align*} 
Therefore, we can assume \eqref{asm:term:1} by assuming $\Bigg\|\frac{1}{T-u(T)-1}\sum^{T}_{t=u(T)+1}\pi_t(a\mid X, \Omega_{t-1}) - \bar{\alpha}(a\mid X)\Bigg\|_2=\op(T^{-p})$ and $\Big\|\bar{\alpha}(a\mid X) - \hat{g}_{t-1}(a\mid X)\Big\|_2=\op(T^{-p})$. Then, we have $\eqref{eq:con1} = \op(T^{-p})\op(u(T)^{-q}) = \op(T^{-p})\op((T^{1/3})^{-q})=\op(T^{-p}(T^{1/3})^{-q})=\op(T^{-(p+q/3)})$. In this case, for instance, if $p=7/18 (< 1/2)$ and $q=1/3$, we obtain $\eqref{eq:con1} = \op(T^{-1/2})$. Similarly, we have $\eqref{eq:con2} =  \op(u(T)^{-q})\op(t^{-r}) = \op(T^{-q/3})\frac{1}{T}\sum^T_{t=1}\op(t^{-r})$. By a property of Riemann Zeta function, $\frac{1}{T}\sum^T_{t=1}\op(t^{-r}) = \op(T^{-r})$ Hence, if $r=7/18$ and $q=1/3$, we can show that $\eqref{eq:con2} = \op(T^{-r-q/3}) = \op(T^{-1/2})$.

\section{Monte Carlo Experiments}
\label{sec:exp}
To investigate the empirical properties of ADR and Modified ADR (MADR) estimators, we simulate two situations based on whether the logging policy converges, where the average logging policy converges in all experiments. We compare ADR and MADR estimators with an IPW estimator with the true logging policy (IPW), IPW estimator with an estimated logging policy (EIPW), AIPW estimator without cross fitting (AIPW), DM estimator (DM), DR estimator without cross fitting (DR), and A2IPW estimator with the true logging policy (A2IPW). We also consider estimators with the following form:
\begin{align*}
\frac{1}{T}\sum^T_{t=1}\sum^K_{a=1}\Bigg\{\frac{\epol(a\mid X_t)\mathbbm{1}[A_t=a]\left(Y_t - \hat{f}(a, X_t)\right) }{\frac{1}{t}\sum^{t}_{s=1}\pi_{s}(a\mid X_t, \Omega_{s-1})} + \epol(a\mid X_t)\hat{f}(a, X_t)\Bigg\}.
\end{align*}
When using $\hat{f} = \hat{f}_{t-1}$, we call it Average A2IPW (A3IPW) estimator; when using $\hat{f} = \tilde{f}_{t-1, u(T)}$, we call it an Modified A3IPW (MA3IPW) estimator. These estimators are special cases of ADR and MADR estimators where using $\frac{1}{t}\sum^{t}_{s=1}\pi_{s}(a\mid X_t, \Omega_{s-1})$ for $\hat{g}_{t-1}$. Note that only the ADR, MADR, EIPW, and DM estimators are applicable even when the true logging policy is unknown among these estimators. In addition, to the best of our knowledge, the EIPW, AIPW, and DR estimators are not shown to be asymptotically normal. When the logging policy does not converge to a time-invariant function, the A2IPW estimator also does not have the asymptotic normality. When the average logging policy converges and convergence rate conditions hold, the MADR and MA3IPW estimators have asymptotic normality. For the asymptotic normalities of ADR and MADR estimator, we need Assumption~\ref{asm:mean_stat}, which is not easy to be confirmed.

\begin{figure}[t]
\begin{center}
 \includegraphics[width=135mm]{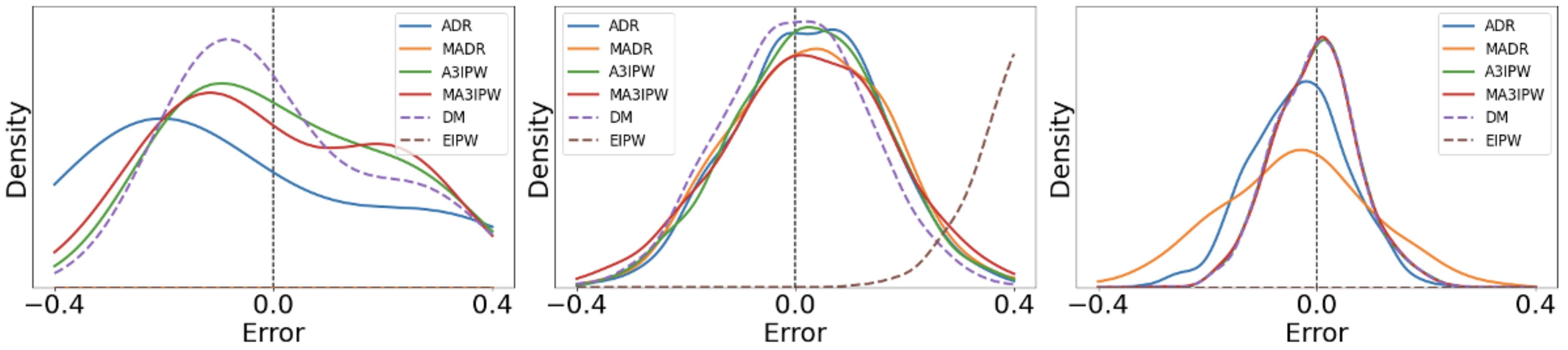}
\end{center}
\caption{This figure illustrates the error distributions of estimators in Section~\ref{sec:exp_adaptive}. We smoothed the error distributions using kernel density estimation.}
\label{fig:synthetic}
\end{figure}

\subsection{Experiments with a Fluctuating logging policy}
\label{sec:exp_adaptive}
\begin{table}[t]
\caption{Experimental results for dependent samples in adaptive efficient ATE estimation. The upper table shows the RMSEs and SDs and the lower table shows the coverage ratios of the confidence interval. We show the RMSEs, SDs, and CRs.} 
\label{tbl:exp_table_moe}
\begin{center}
\scalebox{0.9}[0.9]{
\begin{tabular}{|cl|r|r|r|}
\hline
\multicolumn{2}{|c|}{Sample size $T$} & 250 &  500 &750 \\
\hline
\multirow{3}{*}{DM} & RMSE &  0.164 &  0.132 &  0.076 \\
{} & SD  &  0.030 &  0.025 &  0.008 \\
{} & CR  &  0.000 &  0.000 &  0.000 \\
\hline
\multirow{3}{*}{EIPQ} & RMSE &  1.332 &  0.514 &  1.334 \\
{} & SD  &  0.147 &  0.120 &  0.074 \\
{} & CR  &  0.000 &  0.039 &  0.000 \\
\hline
\multirow{3}{*}{ADR} & RMSE &  0.245 &  0.138 &  0.095 \\
{} & SD  &  0.037 &  0.027 &  0.013 \\
{} & CR  &  0.125 &  0.926 &  0.452 \\
\hline
\multirow{3}{*}{MADR} & RMSE &  1.215 &  0.148 &  0.146 \\
{} & SD &  0.151 &  0.031 &  0.030 \\
{} & CR &  0.000 &  0.908 &  0.332 \\
\hline
\multirow{3}{*}{A3IPW} & RMSE &  0.173 &  0.141 &  0.077 \\
{} & SD &  0.028 &  0.028 &  0.009 \\
{} & CR &  1.000 &  0.973 &  0.952 \\
\hline
\multirow{3}{*}{MA3IPW}  & RMSE &  0.180 &  0.158 &  0.077 \\
{} & SD &  0.022 &  0.036 &  0.009 \\
{} & CR &  1.000 &  0.949 &  0.952 \\
\hline
\end{tabular}
} 
\end{center}
\vspace{-0.5cm}
\end{table}

In this section, we conduct adaptive experiments for efficient ATE estimation following \citet{Laan2008TheCA} and \citet{Hahn2011}. For brevity, we consider a situation where there are two actions and no covariates; that is, there is no sample selection bias based on $p_t(a\mid x)$. We generate a pair of potential outcomes $(Y_t(1), Y_t(2),)$, where $Y_t(a)$ is generated from the normal distribution $\mathcal{N}(a, a)$. Let us define an ATE by defining an evaluation weight as $\epol(1) = -1$ and $\epol(2) = 1$. \citet{Laan2008TheCA} and \citet{Hahn2011} showed that we can achieve the minimum asymptotic variance when choosing an action $1$ following a probability $\pi^*(1) = \frac{\sqrt{\mathrm{Var}(Y_t(1))}}{\sqrt{\mathrm{Var}(Y_t(1))} + \sqrt{\mathrm{Var}(Y_t(2))}}$ and the other action following $\pi^*(2) = 1-\pi^*(1)$. However, because we do not know $\mathrm{Var}(Y_t(a))$, we need to consider obtaining an estimator with the same asymptotic distribution as the one obtained under an optimal logging policy $\pi^*$. In this paper, we select an action with probability $1$ so that the ratio of $\sum^T_{t=1}\pi_1(1) = \sum^T_{t=1}\mathbbm{1}[A_t = 1]$ and  $\sum^T_{t=1}\pi_1(2) = \sum^T_{t=1}\mathbbm{1}[A_t = 2]$ is $\frac{\sqrt{\mathrm{Var}(Y_t(1))}}{\sqrt{\mathrm{Var}(Y_t(1))} + \sqrt{\mathrm{Var}(Y_t(2))}} : \frac{\sqrt{\mathrm{Var}(Y_t(2))}}{\sqrt{\mathrm{Var}(Y_t(1))} + \sqrt{\mathrm{Var}(Y_t(2))}}$. If the average logging policy converges to $\tilde{\alpha}(a) = \pi^*(a)$, the asymptotic distribution of and ADR estimator is the same as that of an estimator obtained when choosing an action $a$ with a probability $\pi^*(a)$. 
To keep the desirable ratio, at each period $t$, we estimate the standard deviation $\mathrm{Var}(Y_t(a))$ using $\Omega_{t-1}$. Next, we construct an estimator $\hat{\pi}^*(a)$ of $\pi^*(a)$. If $\sum^T_{t=1}\mathbbm{1}[A_t = 1]\leq \hat{\pi}^*(a)$, we choose $A_t = 1$; otherwise, $A_t = 2$. We conduct this procedure for three cases with different sample sizes $T=250, 500, 750$. We conduct $100$ trials and calculate the root MSEs (RMSEs), the standard deviations of MSEs (SDs), and the coverage ratios (CRs) of the $95\%$ confidence interval; that is, a percentage that the confidence interval covers the true value. The results are shown in Table~\ref{tbl:exp_table_moe} and Figure~\ref{fig:synthetic}. These results imply that the proposed estimators successfully estimate the mean outcome, although the EIPW estimator shows significantly bad performance. The DM estimator seems to estimate it well, but the confidence interval does not work, as the coverage ratio shows. In addition, in this experiment, we assume that there is no covariate for brevity. However, as the experiment of Section~\ref{sec:conv_logg_prob}, if there is a sample selection bias owing to a covariate-dependent logging policy, the DM estimator's performance relatively decreases. The ADR, MADR, A3IPW, and MA3IPW estimators have the same asymptotic  distribution under appropriate conditions. Although the performances of A2IPW and MA3IPW are superior to ADR and MADR estimators, the estimators are applicable only when the true logging policy is known. 

\subsection{Experiments with a Converging logging policy}
\label{sec:conv_logg_prob}
Next, we investigate the performances of the estimators for logging policies that converge to a time-invariant function. We generate an artificial pair of covariate and potential outcome $(X_t, Y_t(1), Y_t(2), Y_t(3))$. The covariate $X_t$ is a $10$ dimensional vector generated from the standard normal distribution. For $a\in\{1,2,3\}$, the potential outcome $Y_t(a)$ is $1$ if $a$ is chosen by following a probability defined as $p(a\mid x) = \frac{\exp(g(a, x))}{\sum^3_{a'}\exp(g(a', x))}$, where $g(1, x) = \sum^{10}_{d=1} X_{t,d}$, $g(2, x) = \sum^{10}_{d=1} W_dX^2_{t,d}$, and $g(3, x) = \sum^{10}_{d=1} W_d|X_{t,d}|$, where $W_d$ is uniform randomly chosen from $\{-1, 1\}$. Let us generate three datasets, $\mathcal{S}^{(1)}_{T_{(1)}}$, $\mathcal{S}^{(2)}_{T_{(2)}}$, and $\mathcal{S}^{(3)}_{T_{(3)}}$, where $\mathcal{S}^{(m)}_{T_{(m)}} = \{(X^{(m)}_t, Y^{(m)}_t(1), Y^{(m)}_t(2), Y^{(m)}_t(3))\}^{T_{(m)}}_{t=1}$. Firstly, we train an logging policy $\epol$ by solving a prediction problem between $X^{(1)}_t$ and $Y^{(1)}_t(1), Y^{(1)}_t(2), Y^{(1)}_t(3)$ using the dataset $\mathcal{S}^{(1)}_{T_{(1)}}$. Then, we apply the evaluation policy $\epol$ on the independent dataset $\mathcal{S}^{(2)}_{T_{(2)}}$, and artificially construct bandit data $\{(X'_t, A'_t, Y'_t)\}^{T_{(2)}}_{t=1}$, where $A'_t$ is a chosen action from the evaluation policy and $Y^{(m)}_t = \sum^3_{a=1}\mathbbm{1}[A^{(m)}_t = a]Y^{(m)}_t(a)$. Then, we set the true policy value $R(\epol)$ as $\frac{1}{T_{(2)}}\sum^{T_{(2)}}_{t=1}Y^{(m)}_t$. Next, using the datasets $\mathcal{S}^{(3)}_{T_{(3)}}$ and a MAB algorithm, we generate a bandit dataset as $\mathcal{S}=\{(X_t, A_t, Y_t)\}^{T_{(3)}}_{t=1}$. For the dataset $\mathcal{S}$, we apply the IPW estimator with the true logging policy, IPW estimator with estimated logging policy, AIPW estimator with cross fitting, DM estimator, DR estimator with cross fitting, A2IPW estimator, and ADR estimator. For estimating $\hat{f}$ and $\hat{g}$, we use the kernelized Ridge least squares and kernelized Ridge logistic regression, respectively. We use the Gaussian kernel for the kernel, and the hyper-parameters of the regularization and the kernel are chosen from $\{0.01, 0.1, 1\}$. Let us define an estimation error as $R(\epol) - \widehat{R}(\epol)$. We conduct six experiments by changing the sample size and the MAB algorithms. For the sample size $T_{(3)}$, we use $250$, $500$, and $750$. For each sample size, we apply the LinUCB and LinTS algorithms. For the sample size $T_{(1)}$ and $T_{(2)}$, we use $1,000$ and $100,000$, respectively. For $100$ trials, we show the average root MSEs (RMSEs), the standard deviations of MSEs (SDs), and the coverage ratios (CRs) of the $95\%$ confidence interval. The results are shown in Table~\ref{tbl:exp_table1} and Figure~\ref{fig:UCBTS}. Among the estimators, the DR and ADR estimators achieve lower MSEs well, but only the ADR estimator is shown to be asymptotically normal when samples are dependent. Although the IPW estimator returns confidence intervals with a coverage ratio near $95\%$, the MSE is larger than DR-type estimators. Unlike the experiments with a fluctuating logging policy, the DM estimator does not perform owing to the estimation error of $f^*$. Note that the previous experiment, there is no covariate and no sample selection bias caused by $p_t(a\mid x)$. Hence, in this case, it is easy to estimate $f^*$. Here note that the A3IPW and MA3PIPW estimators require the true logging policy, unlike the ADR and MADR estimators.

\begin{figure}[t]
\begin{center}
 \includegraphics[width=138mm]{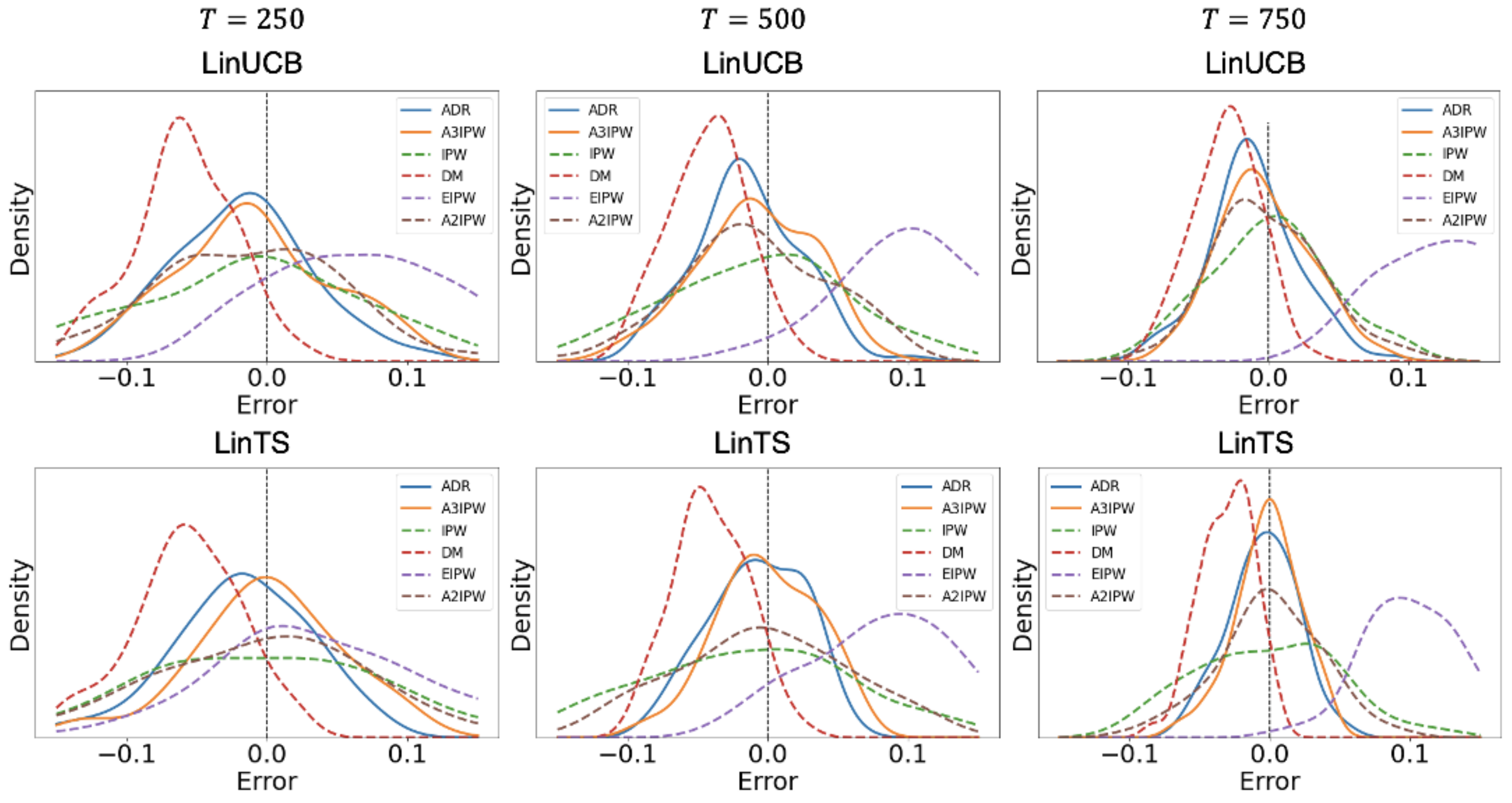}
\end{center}
\vspace{-0.3cm}
\caption{This figure illustrates the error distributions of estimators from dependent samples generated from LinUCB and LinTS policies with the sample size $250$ (left graphs), $500$  (center graphs), and $750$  (right graphs). The upper graphs show the error distributions of the LinUCB policy; the lower graphs show  the error distributions of the LinTS policy. We smoothed the error distributions using kernel density estimation.}
\label{fig:UCBTS}
\vspace{-0.3cm}
\end{figure} 

\begin{table}[t]
\caption{Experimental results of mean outcome estimation from dependent sample generated from LinUCB and LinTS policies with the sample size $500$. We show the RMSEs, SDs, and CRs.} 
\label{tbl:exp_table1}
\begin{center}
\scalebox{0.9}[0.9]{
\begin{tabular}{|cl|rr|rr|rr|}
\hline
\multicolumn{2}{|c|}{Sample size $T$} & \multicolumn{2}{c|}{250} &  \multicolumn{2}{c|}{500} & \multicolumn{2}{c|}{750} \\
\multicolumn{2}{|c|}{MAB algorithm} & LinUCB & LinTS & LinUCB & LinTS & LinUCB & LinTS\\
\hline
\multirow{3}{*}{IPW} & RMSE  &  0.080 &  0.095 &  0.064 &  0.073 &  0.043 &  0.050 \\
{} & SD  &  0.006 &  0.009 &  0.004 &  0.005 &  0.002 &  0.003 \\
{} & CR  &  0.920 &  0.910 &  0.880 &  0.860 &  0.980 &  0.950 \\
\hline
\multirow{3}{*}{DM} & RMSE  &  0.068 &  0.069 &  0.051 &  0.048 &  0.038 &  0.036 \\
{} & SD  &  0.005 &  0.005 &  0.003 &  0.002 &  0.002 &  0.001 \\
{} & CR  &  0.160 &  0.200 &  0.120 &  0.190 &  0.220 &  0.170 \\
\hline
\multirow{3}{*}{AIPW} & RMSE  &  0.056 &  0.067 &  0.048 &  0.054 &  0.036 &  0.033 \\
{} & SD  &  0.004 &  0.006 &  0.003 &  0.004 &  0.002 &  0.002 \\
{} & CR  &  0.940 &  0.910 &  0.930 &  0.870 &  0.910 &  0.980 \\
\hline
\multirow{3}{*}{A2IPW} & RMSE  &  0.066 &  0.082 &  0.052 &  0.061 &  0.038 &  0.037 \\
{} & SD &  0.006 &  0.010 &  0.004 &  0.004 &  0.002 &  0.002 \\
{} & CR &  0.920 &  0.930 &  0.930 &  0.880 &  0.960 &  0.930 \\
\hline
\multirow{3}{*}{EIPW} & RMSE &  0.093 &  0.082 &  0.118 &  0.098 &  0.131 &  0.107 \\
{} & SD &  0.010 &  0.010 &  0.013 &  0.009 &  0.012 &  0.007 \\
{} & CR &  0.770 &  0.860 &  0.330 &  0.530 &  0.130 &  0.170 \\
\hline
\multirow{3}{*}{DM} & RMSE &  0.046 &  0.045 &  0.038 &  0.033 &  0.031 &  0.023 \\
{} & SD &  0.003 &  0.004 &  0.002 &  0.001 &  0.001 &  0.001 \\
{} & CR &  0.930 &  0.910 &  0.890 &  0.960 &  0.920 &  0.960 \\
\hline
\multirow{3}{*}{ADR} & RMSE &  0.052 &  0.052 &  0.039 &  0.034 &  0.033 &  0.025 \\
{} & SD &  0.004 &  0.004 &  0.002 &  0.001 &  0.002 &  0.001 \\
{} & CR &  0.980 &  0.950 &  0.930 &  0.970 &  0.940 &  0.990 \\
\hline
\multirow{3}{*}{MADR} & RMSE &  0.100 &  0.080 &  0.117 &  0.114 &  0.119 &  0.119 \\
{} & SD &  0.016 &  0.009 &  0.015 &  0.018 &  0.015 &  0.014 \\
{} & CR &  0.800 &  0.890 &  0.470 &  0.540 &  0.360 &  0.290 \\
\hline
\multirow{3}{*}{A3IPW} & RMSE &  0.055 &  0.053 &  0.042 &  0.035 &  0.032 &  0.023 \\
{} & SD &  0.004 &  0.004 &  0.002 &  0.002 &  0.001 &  0.001 \\
{} & CR &  0.980 &  0.940 &  0.930 &  0.970 &  0.940 &  0.970 \\
\hline
\multirow{3}{*}{MA3IPW} & RMSE &  0.059 &  0.061 &  0.048 &  0.054 &  0.039 &  0.044 \\
{} & SD &  0.005 &  0.005 &  0.004 &  0.005 &  0.002 &  0.002 \\
{} & CR &  0.970 &  0.940 &  0.930 &  0.870 &  0.930 &  0.850 \\
\hline
\end{tabular}
} 
\end{center}
\vspace{-0.5cm}
\end{table}

\section{Discussion}
\label{sec:discuss}
We discuss the remaining problems. First, we consider a paradox of using an estimated logging policy. \citet{hadad2019} pointed out an A2IPW estimator's unstable behavior using samples obtained from a MAB algorithm. On the other hand, \citet{kato2020theoreticalcomparison} pointed out that the ADR estimator experimentally shows better performance than the A2IPW estimator even though their asymptotic properties are the same. This paper points out that estimating the logging policy is equivalent to estimating the average logging policy. In our experiment, directly using the average logging policy for an A2IPW estimator also improves the performance. Therefore, we conjecture that an ADR estimator's stabilization effect comes from the stability of the average logging policy.

Unlike an A2IPW estimator with the true logging policy, an ADR estimator does not suffer the deficient support problem \citep{Sachdeva2020}. In many cases of MAB algorithms, the logging policy $\pi_t$ often becomes $0$. However, even if $\pi_t$ becomes $0$, we can show the asymptotic normality under Assumption~\ref{asm:mean_stat}.

Next, in addition to the double robustness on the consistency, we explain the importance of the DR-type estimators' form. We showed the asymptotic normality only for the DR-type estimator. Readers may feel that we can show the other type estimators' asymptotic normality, such as an IPW-type estimator. However, it is not obvious how we use this paper's inference strategy to such estimators. As \citet{ChernozhukovVictor2018Dmlf} discussed, the DR-type estimators relax the condition for asymptotic normality when using sample-splitting. The asymptotic normalities of our proposed estimators are also based on this property. Therefore, the form of the DR-type estimators is also essential. 

Finally, \citet{Luedtke2016} pointed out that it is difficult to show the asymptotic normality when using a non-unique optimal treatment strategy; that is, $\pi_t$ fluctuates and does not converge. This problem is also partially solved in our proposed method if Assumption~\ref{asm:mean_stat} holds. For instance, in BAI, \citet{kaufman2016complexity} proposed an algorithm that they deterministically choose an arm with a probability $1$ to keep some optimal selection ratio of arms. In this case, in addition to the deficient support problem, there is no unique treatment strategy. However, because the algorithm attempts to keep some desirable ratio, we can apply our method under Assumption~\ref{asm:mean_stat}. We derived the asymptotically normal mean outcome estimator for dependent samples based on a new assumption that the average logging policy converges to a time-invariant function in probability. We can regard the average logging policy as a propensity score of inverse probability weighting under this setting. On the other hand, existing studies need to assume that the logging policy itself converges and use the logging policy as a propensity score. We also experimentally confirmed that the inverse weighting using the average logging policy is more stable than the A2IPW estimator, which uses the true logging policy itself.

\section{Conclusion}
We derived the asymptotically normal mean outcome estimator for dependent samples based on a new assumption that the average logging policy converges to a time-invariant function in probability. We can regard the average logging policy as a propensity score of inverse probability weighting under this setting. On the other hand, existing studies need to assume that the logging policy itself converges and use the logging policy as a propensity score. We also experimentally confirmed that the inverse weighting using the average logging policy is more stable than using the logging policy itself. 

\appendix

\bibliographystyle{chicago}
\bibliography{MeanStatOPE}

\onecolumn
\appendix

\section{Preliminaries}
\label{appdx:prelim}

\begin{definition}\label{dfn:uniint}[Uniformly Integrable, \citet{GVK126800421}, p.~191]  A sequence $\{A_t\}$ is said to be uniformly integrable if for every $\epsilon > 0$ there exists a number $c>0$ such that 
\begin{align*}
\mathbb{E}[|A_t|\cdot I[|A_t \geq c|]] < \epsilon
\end{align*}
for all $t$.
\end{definition}

\begin{proposition}\label{prp:suff_uniint}[Sufficient Conditions for Uniformly Integrable, \citet{GVK126800421}, Proposition~7.7, p.~191]  (a) Suppose there exist $r>1$ and $M<\infty$ such that $\mathbb{E}[|A_t|^r]<M$ for all $t$. Then $\{A_t\}$ is uniformly integrable. (b) Suppose there exist $r>1$ and $M < \infty$ such that $\mathbb{E}[|b_t|^r]<M$ for all $t$. If $A_t = \sum^\infty_{j=-\infty}h_jb_{t-j}$ with $\sum^\infty_{j=-\infty}|h_j|<\infty$, then $\{A_t\}$ is uniformly integrable.
\end{proposition}

\begin{proposition}[$L^r$ Convergence Theorem, \citet{loeve1977probability}]
\label{prp:lr_conv_theorem}
Let $0<r<\infty$, suppose that $\mathbb{E}\big[|a_n|^r\big] < \infty$ for all $n$ and that $a_n \xrightarrow{\mathrm{p}}a$ as $n\to \infty$. The following are equivalent: 
\begin{description}
\item{(i)} $a_n\to a$ in $L^r$ as $n\to\infty$;
\item{(ii)} $\mathbb{E}\big[|a_n|^r\big]\to \mathbb{E}\big[|a|^r\big] < \infty$ as $n\to\infty$; 
\item{(iii)} $\big\{|a_n|^r, n\geq 1\big\}$ is uniformly integrable.
\end{description}
\end{proposition}

\begin{proposition}
\label{prp:mrtgl_WLLN}[Weak Law of Large Numbers for Martingale, \citet{hall2014martingale}]
Let $\{S_n = \sum^{n}_{i=1} X_i, \mathcal{H}_{t}, t\geq 1\}$ be a martingale and $\{b_n\}$ a sequence of positive constants with $b_n\to\infty$ as $n\to\infty$. Then, writing $X_{ni} = X_i\mathbbm{1}[|X_i|\leq b_n]$, $1\leq i \leq n$, we have that $b^{-1}_n S_n \xrightarrow{\mathrm{p}} 0$ as $n\to \infty$ if 
\begin{description}
\item[(i)] $\sum^n_{i=1}P(|X_i| > b_n)\to 0$;
\item[(ii)] $b^{-1}_n\sum^n_{i=1}\mathbb{E}[X_{ni}\mid \mathcal{H}_{t-1}] \xrightarrow{\mathrm{p}} 0$, and;
\item[(iii)] $b^{-2}_n \sum^n_{i=1}\big\{\mathbb{E}[X^2_{ni}] - \mathbb{E}\big[\mathbb{E}\big[X_{ni}\mid \mathcal{H}_{t-1}\big]\big]^2\big\}\to 0$.
\end{description}
\end{proposition}
\begin{remark} The weak law of large numbers for martingale holds when the random variable is bounded by a constant.
\end{remark}

\section{Proof of Lemma~\ref{LEM:1}}
\label{appdx:proof:target1}
We prove Lemma~\ref{LEM:1} by using a similar technique used in \citet{Laan2014onlinetml} and Theorem~1 of \citet{kato2020theoreticalcomparison}. 

\begin{proof}
Let us define 
\begin{align*}
&\phi_1(X_t, A_t, Y_t; g, f)=\sum^K_{a=1}\frac{\epol(a\mid X_t)\mathbbm{1}[A_t=a]\left(Y_t - f(a, X_t)\right) }{g(a\mid X_t)},\\
&\phi_2(X_t; f)=\sum^K_{a=1}\epol(a\mid X_t)f(a, X_t).
\end{align*}
We decompose $\sqrt{T}\Big(\widehat{R}^{\mathrm{ADR}}_T(\epol) - \ddot{R}_T(\epol)\Big)$ as 
\begin{align*}
&\widehat{R}^{\mathrm{ADR}}_T(\epol) - \ddot{R}(\epol)\\
&=\frac{1}{T}\sum^T_{t=1}\Bigg\{\phi_1(X_t, A_t, Y_t; \hat{g}_{t-1}, \hat{f}_{t-1}) - \phi_1(X_t, A_t, Y_t; \bar{\alpha}, f^*)\\
&\ \ \ -\mathbb{E}\left[\phi_1(X_t, A_t, Y_t; \hat{g}_{t-1}, \hat{f}_{t-1}) - \phi_1(X_t, A_t, Y_t; \bar{\alpha}, f^*)\mid \Omega_{t-1}\right]\\
&\ \ \ + \phi_2(X_t; \hat{f}_{t-1}) - \phi_2(X_t; f^*) -\mathbb{E}\left[\phi_2(X_t; \hat{f}_{t-1}) - \phi_2(X_t; f^*)\mid \Omega_{t-1}\right]\Bigg\}\\
&\ \ \ + \frac{1}{T}\sum^T_{t=1}\mathbb{E}\left[\phi_1(X_t, A_t, Y_t; \hat{g}_{t-1}, \hat{f}_{t-1})\mid \Omega_{t-1}\right] + \frac{1}{T}\sum^T_{t=1}\mathbb{E}\left[\phi_2(X_t; \hat{f}_{t-1})\mid \Omega_{t-1}\right]\\
&\ \ \ - \frac{1}{T}\sum^T_{t=1}\mathbb{E}\left[\phi_1(X_t, A_t, Y_t; \bar{\alpha}, f^*)\mid \Omega_{t-1}\right] - \frac{1}{T}\sum^T_{t=1}\mathbb{E}\left[\phi_2(X_t; f^*)\mid \Omega_{t-1}\right].
\end{align*}
In the following parts, we separately show that
\begin{align}
\label{eq:part1}
&\sqrt{T}\frac{1}{T}\sum^T_{t=1}\Bigg\{\phi_1(X_t, A_t, Y_t; \hat{g}_{t-1}, \hat{f}_{t-1}) - \phi_1(X_t, A_t, Y_t; \bar{\alpha}, f^*)\\
&\ \ \ -\mathbb{E}\left[\phi_1(X_t, A_t, Y_t; \hat{g}_{t-1}, \hat{f}_{t-1}) - \phi_1(X_t, A_t, Y_t; \bar{\alpha}, f^*)\mid \Omega_{t-1}\right]\nonumber\\
&\ \ \ + \phi_2(X_t; \hat{f}_{t-1}) - \phi_2(X_t; f^*) -\mathbb{E}\left[\phi_2(X_t; \hat{f}_{t-1}) - \phi_2(X_t; f^*)\mid \Omega_{t-1}\right]\Bigg\}\nonumber\\
&= \mathrm{o}_p(1);\nonumber
\end{align}
and 
\begin{align}
\label{eq:part2}
&\frac{1}{T}\sum^T_{t=1}\mathbb{E}\left[\phi_1(X_t, A_t, Y_t; \hat{g}_{t-1}, \hat{f}_{t-1})\mid \Omega_{t-1}\right] + \frac{1}{T}\sum^T_{t=1}\mathbb{E}\left[\phi_2(X_t; \hat{f}_{t-1})\mid \Omega_{t-1}\right]\\
&- \frac{1}{T}\sum^T_{t=1}\mathbb{E}\left[\phi_1(X_t, A_t, Y_t; \bar{\alpha}, f^*)\mid \Omega_{t-1}\right] - \frac{1}{T}\sum^T_{t=1}\mathbb{E}\left[\phi_2(X_t; f^*)\mid \Omega_{t-1}\right] = \mathrm{o}_p(1/\sqrt{T}).\nonumber
\end{align}

\subsection*{Step~1: Proof of \eqref{eq:part1}}

For any $\varepsilon > 0$, to show that 
\begin{align*}
&\mathbb{P}\Bigg(\Bigg|\sqrt{T}\frac{1}{T}\sum^T_{t=1}\Bigg\{\phi_1(X_t, A_t, Y_t; \hat{g}_{t-1}, \hat{f}_{t-1}) - \phi_1(X_t, A_t, Y_t; \bar{\alpha}, f^*)\\
&\ \ \ -\mathbb{E}\left[\phi_1(X_t, A_t, Y_t; \hat{g}_{t-1}, \hat{f}_{t-1}) - \phi_1(X_t, A_t, Y_t; \bar{\alpha}, f^*)\mid \Omega_{t-1}\right]\\
&\ \ \ + \phi_2(X_t; \hat{f}_{t-1}) - \phi_2(X_tt; f^*) -\mathbb{E}\left[\phi_2(X_t; \hat{f}_{t-1}) - \phi_2(X_t; f^*)\mid \Omega_{t-1}\right]\Bigg\}\Bigg| > \varepsilon \Bigg)\\
& \to 0,
\end{align*}
we show that the mean is $0$ and the variance of the component converges to $0$. Then, from the Chebyshev's inequality, this result yields the statement. The mean is 
\begin{align*}
&\sqrt{T}\frac{1}{T}\sum^T_{t=1}\mathbb{E}\Bigg[\Bigg\{\phi_1(X_t, A_t, Y_t; \hat{g}_{t-1}, \hat{f}_{t-1}) - \phi_1(X_t, A_t, Y_t; \bar{\alpha}, f^*)\\
&\ \ \ \ \ \ \ -\mathbb{E}\left[\phi_1(X_t, A_t, Y_t; \hat{g}_{t-1}, \hat{f}_{t-1}) - \phi_1(X_t, A_t, Y_t; \bar{\alpha}, f^*)\mid \Omega_{t-1}\right]\\
&\ \ \ \ \ \ \ + \phi_2(X_t; \hat{f}_{t-1}) - \phi_2(X_t; f^*) -\mathbb{E}\left[\phi_2(X_t; \hat{f}_{t-1}) - \phi_2(X_t; f^*)\mid \Omega_{t-1}\right]\Bigg\}\Bigg]\\
&=\sqrt{T}\frac{1}{T}\sum^T_{t=1}\mathbb{E}\Bigg[\mathbb{E}\Bigg[\Bigg\{\phi_1(X_t, A_t, Y_t; \hat{g}_{t-1}, \hat{f}_{t-1}) - \phi_1(X_t, A_t, Y_t; \bar{\alpha}, f^*)\\
&\ \ \ \ \ \ \ -\mathbb{E}\left[\phi_1(X_t, A_t, Y_t; \hat{g}_{t-1}, \hat{f}_{t-1}) - \phi_1(X_t, A_t, Y_t; \bar{\alpha}, f^*)\mid \Omega_{t-1}\right]\\
&\ \ \ \ \ \ \ + \phi_2(X_t; \hat{f}_{t-1}) - \phi_2(X_t; f^*) -\mathbb{E}\left[\phi_2(X_t; \hat{f}_{t-1}) - \phi_2(X_t; f^*)\mid \Omega_{t-1}\right]\Bigg\}\mid \Omega_{t-1}\Bigg]\Bigg]\\
& = 0
\end{align*}

Because the mean is $0$, the variance is
\begin{align*}
&\mathrm{Var}\Bigg(\sqrt{T}\frac{1}{T}\sum^T_{t=1}\Bigg\{\phi_1(X_t, A_t, Y_t; \hat{g}_{t-1}, \hat{f}_{t-1}) - \phi_1(X_t, A_t, Y_t; \bar{\alpha}, f^*)\\
&\ \ \ \ \ \ \ -\mathbb{E}\left[\phi_1(X_t, A_t, Y_t; \hat{g}_{t-1}, \hat{f}_{t-1}) - \phi_1(X_t, A_t, Y_t; \bar{\alpha}, f^*)\mid \Omega_{t-1}\right]\\
&\ \ \ \ \ \ \ + \phi_2(X_t; \hat{f}_{t-1}) - \phi_2(X_t; f^*) -\mathbb{E}\left[\phi_2(X_t; \hat{f}_{t-1}) - \phi_2(X_t; f^*)\mid \Omega_{t-1}\right]\Bigg\}\Bigg)\\
&=\mathbb{E}\Bigg[\Bigg(\sqrt{T}\frac{1}{T}\sum^T_{t=1}\Bigg\{\phi_1(X_t, A_t, Y_t; \hat{g}_{t-1}, \hat{f}_{t-1}) - \phi_1(X_t, A_t, Y_t; \bar{\alpha}, f^*)\\
&\ \ \ \ \ \ \ -\mathbb{E}\left[\phi_1(X_t, A_t, Y_t; \hat{g}_{t-1}, \hat{f}_{t-1}) - \phi_1(X_t, A_t, Y_t; \bar{\alpha}, f^*)\mid \Omega_{t-1}\right]\\
&\ \ \ \ \ \ \ + \phi_2(X_t; \hat{f}_{t-1}) - \phi_2(X_t; f^*) -\mathbb{E}\left[\phi_2(X_t; \hat{f}_{t-1}) - \phi_2(X_t; f^*)\mid \Omega_{t-1}\right]\Bigg\}\Bigg)^2\Bigg]\\
&=\frac{1}{T}\mathbb{E}\Bigg[\Bigg(\sum^T_{t=1}\Bigg\{\phi_1(X_t, A_t, Y_t; \hat{g}_{t-1}, \hat{f}_{t-1}) - \phi_1(X_t, A_t, Y_t; \bar{\alpha}, f^*)\\
&\ \ \ \ \ \ \ -\mathbb{E}\left[\phi_1(X_t, A_t, Y_t; \hat{g}_{t-1}, \hat{f}_{t-1}) - \phi_1(X_t, A_t, Y_t; \bar{\alpha}, f^*)\mid \Omega_{t-1}\right]\\
&\ \ \ \ \ \ \ + \phi_2(X_t; \hat{f}_{t-1}) - \phi_2(X_t; f^*) -\mathbb{E}\left[\phi_2(X_t; \hat{f}_{t-1}) - \phi_2(X_t; f^*)\mid \Omega_{t-1}\right]\Bigg\}\Bigg)^2\Bigg].
\end{align*}
Therefore, we have
\begin{align*}
&=\frac{1}{T}\sum^T_{t=1}\mathbb{E}\Bigg[\Bigg(\phi_1(X_t, A_t, Y_t; \hat{g}_{t-1}, \hat{f}_{t-1}) - \phi_1(X_t, A_t, Y_t; \bar{\alpha}, f^*)\\
&\ \ \ \ \ \ \ -\mathbb{E}\left[\phi_1(X_t, A_t, Y_t; \hat{g}_{t-1}, \hat{f}_{t-1}) - \phi_1(X_t, A_t, Y_t; \bar{\alpha}, f^*)\mid \Omega_{t-1}\right]\\
&\ \ \ \ \ \ \ + \phi_2(X_t; \hat{f}_{t-1}) - \phi_2(X_t; f^*) -\mathbb{E}\left[\phi_2(X_t; \hat{f}_{t-1}) - \phi_2(X_t; f^*)\mid \Omega_{t-1}\right]\Bigg)^2\Bigg]\\
&\ \ \ +\frac{2}{T}\sum^{T-1}_{t=1}\sum^T_{s=t+1}\mathbb{E}\Bigg[\Bigg(\phi_1(X_t, A_t, Y_t; \hat{g}_{t-1}, \hat{f}_{t-1}) - \phi_1(X_t, A_t, Y_t; \bar{\alpha}, f^*)\\
&\ \ \ \ \ \ \ -\mathbb{E}\left[\phi_1(X_t, A_t, Y_t; \hat{g}_{t-1}, \hat{f}_{t-1}) - \phi_1(X_t, A_t, Y_t; \bar{\alpha}, f^*)\mid \Omega_{t-1}\right]\\
&\ \ \ \ \ \ \ + \phi_2(X_t; \hat{f}_{t-1}) - \phi_2(X_t; f^*) -\mathbb{E}\left[\phi_2(X_t; \hat{f}_{t-1}) - \phi_2(X_t; f^*)\mid \Omega_{t-1}\right]\Bigg)\\
&\ \ \ \ \ \ \ \times \Bigg(\phi_1(X_s, A_s, Y_s; \hat{g}_{s-1}, \hat{f}_{s-1}) - \phi_1(X_s, A_s, Y_s; \bar{\alpha}, f^*)\\
&\ \ \ \ \ \ \ -\mathbb{E}\left[\phi_1(X_s, A_s, Y_s; \hat{g}_{s-1}, \hat{f}_{s-1}) - \phi_1(X_s, A_s, Y_s; \bar{\alpha}, f^*)\mid \Omega_{s-1}\right]\\
&\ \ \ \ \ \ \ + \phi_2(X_s; \hat{f}_{s-1}) - \phi_2(X_s; f^*) -\mathbb{E}\left[\phi_2(X_s; \hat{f}_{s-1}) - \phi_2(X_s; f^*)\mid \Omega_{s-1}\right]\Bigg)\Bigg].
\end{align*}

For $s > t$, 
\begin{align*}
&\mathbb{E}\Bigg[\Bigg(\phi_1(X_t, A_t, Y_t; \hat{g}_{t-1}, \hat{f}_{t-1}) - \phi_1(X_t, A_t, Y_t; \bar{\alpha}, f^*)\\
&\ \ \ \ \ \ \ -\mathbb{E}\left[\phi_1(X_t, A_t, Y_t; \hat{g}_{t-1}, \hat{f}_{t-1}) - \phi_1(X_t, A_t, Y_t; \bar{\alpha}, f^*)\mid \Omega_{t-1}\right]\\
&\ \ \ \ \ \ \ + \phi_2(X_t; \hat{f}_{t-1}) - \phi_2(X_t; f^*) -\mathbb{E}\left[\phi_2(X_t; \hat{f}_{t-1}) - \phi_2(X_t; f^*)\mid \Omega_{t-1}\right]\Bigg)\\
&\ \ \ \ \ \ \ \times \Bigg(\phi_1(X_s, A_s, Y_s; \hat{g}_{s-1}, \hat{f}_{s-1}) - \phi_1(X_s, A_s, Y_s; \bar{\alpha}, f^*)\\
&\ \ \ \ \ \ \ -\mathbb{E}\left[\phi_1(X_s, A_s, Y_s; \hat{g}_{s-1}, \hat{f}_{s-1}) - \phi_1(X_s, A_s, Y_s; \bar{\alpha}, f^*)\mid \Omega_{s-1}\right]\\
&\ \ \ \ \ \ \ + \phi_2(X_s; \hat{f}_{s-1}) - \phi_2(X_s; f^*) -\mathbb{E}\left[\phi_2(X_s; \hat{f}_{s-1}) - \phi_2(X_s; f^*)\mid \Omega_{s-1}\right]\Bigg)\Bigg]\\
&=\mathbb{E}\Bigg[U\mathbb{E}\Bigg[\Bigg(\phi_1(X_s, A_s, Y_s; \hat{g}_{s-1}, \hat{f}_{s-1}) - \phi_1(X_s, A_s, Y_s; \bar{\alpha}, f^*)\\
&\ \ \ \ \ \ \ -\mathbb{E}\left[\phi_1(X_s, A_s, Y_s; \hat{g}_{s-1}, \hat{f}_{s-1}) - \phi_1(X_s, A_s, Y_s; \bar{\alpha}, f^*)\mid \Omega_{s-1}\right]\\
&\ \ \ \ \ \ \ + \phi_2(X_s; \hat{f}_{s-1}) - \phi_2(X_s; f^*) -\mathbb{E}\left[\phi_2(X_s; \hat{f}_{s-1}) - \phi_2(X_s; f^*)\mid \Omega_{s-1}\right]\Bigg)\mid \Omega_{s-1}\Bigg]\Bigg]=0,
\end{align*}
where 
\begin{align*}
&U=\Bigg(\phi_1(X_t, A_t, Y_t; \hat{g}_{t-1}, \hat{f}_{t-1}) - \phi_1(X_t, A_t, Y_t; \bar{\alpha}, f^*)\\
&\ \ \ -\mathbb{E}\left[\phi_1(X_t, A_t, Y_t; \hat{g}_{t-1}, \hat{f}_{t-1}) - \phi_1(X_t, A_t, Y_t; \bar{\alpha}, f^*)\mid \Omega_{t-1}\right]\\
&\ \ \  + \phi_2(X_t; \hat{f}_{t-1}) - \phi_2(X_t; f^*) -\mathbb{E}\left[\phi_2(X_t; \hat{f}_{t-1}) - \phi_2(X_t; f^*)\mid \Omega_{t-1}\right]\Bigg).
\end{align*}
Therefore, the variance is calculated as 
\begin{align*}
&\mathrm{Var}\Bigg(\sqrt{T}\frac{1}{T}\sum^T_{t=1}\Bigg\{\phi_1(X_t, A_t, Y_t; \hat{g}_{t-1}, \hat{f}_{t-1}) - \phi_1(X_t, A_t, Y_t; \bar{\alpha}, f^*)\\
&\ \ \ \ \ \ \ -\mathbb{E}\left[\phi_1(X_t, A_t, Y_t; \hat{g}_{t-1}, \hat{f}_{t-1}) - \phi_1(X_t, A_t, Y_t; \bar{\alpha}, f^*)\mid \Omega_{t-1}\right]\\
&\ \ \ \ \ \ \ + \phi_2(X_t; \hat{f}_{t-1}) - \phi_2(X_t; f^*) -\mathbb{E}\left[\phi_2(X_t; \hat{f}_{t-1}) - \phi_2(X_t; f^*)\mid \Omega_{t-1}\right]\Bigg\}\Bigg)\\
&=\frac{1}{T}\sum^T_{t=1}\mathbb{E}\Bigg[\Bigg(\phi_1(X_t, A_t, Y_t; \hat{g}_{t-1}, \hat{f}_{t-1}) - \phi_1(X_t, A_t, Y_t; \bar{\alpha}, f^*)\\
&\ \ \ \ \ \ \ -\mathbb{E}\left[\phi_1(X_t, A_t, Y_t; \hat{g}_{t-1}, \hat{f}_{t-1}) - \phi_1(X_t, A_t, Y_t; \bar{\alpha}, f^*)\mid \Omega_{t-1}\right]\\
&\ \ \ \ \ \ \ + \phi_2(X_t; \hat{f}_{t-1}) - \phi_2(X_t; f^*) -\mathbb{E}\left[\phi_2(X_t; \hat{f}_{t-1}) - \phi_2(X_t; f^*)\mid \Omega_{t-1}\right]\Bigg)^2\Bigg]\\
&=\frac{1}{T}\sum^T_{t=1}\mathbb{E}\Bigg[\mathbb{E}\Bigg[\Bigg(\phi_1(X_t, A_t, Y_t; \hat{g}_{t-1}, \hat{f}_{t-1}) - \phi_1(X_t, A_t, Y_t; \bar{\alpha}, f^*)\\
&\ \ \ \ \ \ \ -\mathbb{E}\left[\phi_1(X_t, A_t, Y_t; \hat{g}_{t-1}, \hat{f}_{t-1}) - \phi_1(X_t, A_t, Y_t; \bar{\alpha}, f^*)\mid \Omega_{t-1}\right]\\
&\ \ \ \ \ \ \ + \phi_2(X_t; \hat{f}_{t-1}) - \phi_2(X_t; f^*) -\mathbb{E}\left[\phi_2(X_t; \hat{f}_{t-1}) - \phi_2(X_t; f^*)\mid \Omega_{t-1}\right]\Bigg)^2\mid \Omega_{t-1}\Bigg]\Bigg]\\
&=\frac{1}{T}\sum^T_{t=1}\mathbb{E}\Bigg[\mathrm{Var}\Bigg(\phi_1(X_t, A_t, Y_t; \hat{g}_{t-1}, \hat{f}_{t-1}) - \phi_1(X_t, A_t, Y_t; \bar{\alpha}, f^*) + \phi_2(X_t; \hat{f}_{t-1}) - \phi_2(X_t; f^*) \mid \Omega_{t-1}\Bigg)\Bigg]\\
&=\frac{1}{T}\sum^T_{t=1}\mathbb{E}\Bigg[\mathrm{Var}\Bigg(\phi_1(X_t, A_t, Y_t; \hat{g}_{t-1}, \hat{f}_{t-1}) - \phi_1(X_t, A_t, Y_t; \bar{\alpha}, f^*) \mid \Omega_{t-1}\Bigg)\Bigg]\\
&\ \ \ +\frac{1}{T}\sum^T_{t=1}\mathbb{E}\Bigg[\mathrm{Var}\Bigg(\phi_2(X_t; \hat{f}_{t-1}) - \phi_2(X_t; f^*) \mid \Omega_{t-1}\Bigg)\Bigg]\\
&\ \ \ +\frac{2}{T}\sum^T_{t=1}\mathbb{E}\Bigg[\mathrm{Cov}\Bigg(\phi_1(X_t, A_t, Y_t; \hat{g}_{t-1}, \hat{f}_{t-1}) - \phi_1(X_t, A_t, Y_t; \bar{\alpha}, f^*), \phi_2(X_t; \hat{f}_{t-1}) - \phi_2(X_t; f^*) \mid \Omega_{t-1}\Bigg)\Bigg].
\end{align*}

Then, we want to show that
\begin{align}
\label{eq:first_e0}
&\frac{1}{T}\sum^T_{t=1}\mathbb{E}\Bigg[\mathrm{Var}\Bigg(\phi_1(X_t, A_t, Y_t; \hat{g}_{t-1}, \hat{f}_{t-1}) - \phi_1(X_t, A_t, Y_t; \bar{\alpha}, f^*) \mid \Omega_{t-1}\Bigg)\Bigg] \to 0,\\
\label{eq:second_e0}
&\frac{1}{T}\sum^T_{t=1}\mathbb{E}\Bigg[\mathrm{Var}\Bigg(\phi_2(X_t; \hat{f}_{t-1}) - \phi_2(X_t; f^*) \mid \Omega_{t-1}\Bigg)\Bigg]  \to 0,\\
\label{eq:third_e0}
&\frac{2}{T}\sum^T_{t=1}\mathbb{E}\Bigg[\mathrm{Cov}\Bigg(\phi_1(X_t, A_t, Y_t; \hat{g}_{t-1}, \hat{f}_{t-1}) - \phi_1(X_t, A_t, Y_t; \bar{\alpha}, f^*), \phi_2(X_t; \hat{f}_{t-1}) - \phi_2(X_t; f^*) \mid \Omega_{t-1}\Bigg)\Bigg]\\
&\ \ \  \to 0
\end{align}

For showing \eqref{eq:first_e0}--\eqref{eq:third_e0}, we consider showing
\begin{align}
\label{eq:first_e}
&\mathrm{Var}\Bigg(\phi_1(X_t, A_t, Y_t; \hat{g}_{t-1}, \hat{f}_{t-1}) - \phi_1(X_t, A_t, Y_t; \bar{\alpha}, f^*) \mid \Omega_{t-1}\Bigg) = \op(1),\\
\label{eq:second_e}
&\mathrm{Var}\Bigg(\phi_2(X_t; \hat{f}_{t-1}) - \phi_2(X_t; f^*) \mid \Omega_{t-1}\Bigg) = \op(1)\\
\label{eq:third_e}
&\mathrm{Cov}\Bigg(\phi_1(X_t, A_t, Y_t; \hat{g}_{t-1}, \hat{f}_{t-1}) - \phi_1(X_t, A_t, Y_t; \bar{\alpha}, f^*), \phi_2(X_t; \hat{f}_{t-1}) - \phi_2(X_t; f^*) \mid \Omega_{t-1}\Bigg) = \op(1),
\end{align}

The first equation \eqref{eq:first_e} is shown as 
\begin{align*}
&\mathrm{Var}\Bigg(\phi_1(X_t, A_t, Y_t; \hat{g}_{t-1}, \hat{f}_{t-1}) - \phi_1(X_t, A_t, Y_t; \bar{\alpha}, f^*) \mid \Omega_{t-1}\Bigg)\\
&\leq \mathbb{E}\Bigg[\Bigg\{\sum^K_{a=1}\frac{\epol(a\mid X_t)\mathbbm{1}[A_t=a]\left(Y_t - \hat{f}_{t-1}(a, X_t)\right) }{\hat{g}_{t-1}(a\mid X_t)}\\
&\ \ \ \ \ \ \ \ - \sum^K_{a=1}\frac{\epol(a\mid X_t)\mathbbm{1}[A_t=a]\left(Y_t - f^*(a, X_t)\right) }{\bar{\alpha}(a\mid X_t)}\Bigg\}^2\mid \Omega_{t-1}\Bigg]\\
&=\mathbb{E}\Bigg[\Bigg\{\sum^K_{a=1}\frac{\epol(a\mid X_t)\mathbbm{1}[A_t=a]\left(Y_t - \hat{f}_{t-1}(a, X_t)\right) }{\hat{g}_{t-1}(a\mid X_t)}\\
&\ \ \ \ \ \ \ \  - \sum^K_{a=1}\frac{\epol(a\mid X_t)\mathbbm{1}[A_t=a]\left(Y_t - f^*(a, X_t)\right) }{\hat{g}_{t-1}(a\mid X_t)}\\
&\ \ \ \ \ \ \ \ + \sum^K_{a=1}\frac{\epol(a\mid X_t)\mathbbm{1}[A_t=a]\left(Y_t - f^*(a, X_t)\right) }{\hat{g}_{t-1}(a\mid X_t)}\\
&\ \ \ \ \ \ \ \ - \sum^K_{a=1}\frac{\epol(a\mid X_t)\mathbbm{1}[A_t=a]\left(Y_t - f^*(a, X_t)\right) }{\bar{\alpha}(a\mid X_t)}\Bigg\}^2\mid \Omega_{t-1}\Bigg]\\
\end{align*}
Here, we have used a parallelogram law from the second to the third equation. Then, we can show that
\begin{align*}
&\leq 2\mathbb{E}\Bigg[\Bigg\{\sum^K_{a=1}\frac{\epol(a\mid X_t)\mathbbm{1}[A_t=a]\left(Y_t - \hat{f}_{t-1}(a, X_t)\right) }{\hat{g}_{t-1}(a\mid X_t)}\\
&\ \ \ \ \ \ \ \ - \sum^K_{a=1}\frac{\epol(a\mid X_t)\mathbbm{1}[A_t=a]\left(Y_t - f^*(a, X_t)\right) }{\hat{g}_{t-1}(a\mid X_t)}\Bigg\}^2\mid \Omega_{t-1}\Bigg]\\
&\ \ \ + 2\mathbb{E}\Bigg[\Bigg\{\sum^K_{a=1}\frac{\epol(a\mid X_t)\mathbbm{1}[A_t=a]\left(Y_t - f^*(a, X_t)\right) }{\hat{g}_{t-1}(a\mid X_t)}\\
&\ \ \ \ \ \ \ \  - \sum^K_{a=1}\frac{\epol(a\mid X_t)\mathbbm{1}[A_t=a]\left(Y_t - f^*(a, X_t)\right) }{\bar{\alpha}(a\mid X_t)}\Bigg\}^2\mid \Omega_{t-1}\Bigg]\\
&\leq 2C\|f^* - \hat{f}_{t-1}\|^2_2 + 2\times 4C\|\hat{g}_{t-1} - \bar{\alpha}\|^2_{2} = \op(1),
\end{align*}
where $C>0$ is a constant.  We have used $|\hat{f}_{t-1}| < C_f$, and $0<\frac{\epol}{\bar{\alpha}} < C_{\bar{\alpha}}$ and Assumption~\ref{asm:consistency}, from the first line to the second inequality. Then, from the $L^r$ convergence theorem (Proposition~\ref{prp:lr_conv_theorem}) and the boundedness of the random variables, we can show that as $t\to \infty$,
\begin{align*}
&\mathbb{E}\Bigg[\mathrm{Var}\Bigg(\phi_1(X_t, A_t, Y_t; \hat{g}_{t-1}, \hat{f}_{t-1}) - \phi_1(X_t, A_t, Y_t; \bar{\alpha}, f^*) \mid \Omega_{t-1}\Bigg)\Bigg]\\
&\leq \mathbb{E}\Bigg[\left|\mathrm{Var}\Bigg(\phi_1(X_t, A_t, Y_t; \hat{g}_{t-1}, \hat{f}_{t-1}) - \phi_1(X_t, A_t, Y_t; \bar{\alpha}, f^*)\mid \Omega_{t-1}\Bigg) \right|\Bigg]\\
&\to 0.
\end{align*}
Therefore, for any $\epsilon > 0$, there exists a constant $ C > 0$ such that 
\begin{align*}
&\frac{1}{T} \sum^{T}_{t=1}\mathbb{E}\Bigg[\mathrm{Var}\Bigg(\phi_1(X_t, A_t, Y_t; \hat{g}_{t-1}, \hat{f}_{t-1}) - \phi_1(X_t, A_t, Y_t; \bar{\alpha}, f^*) \mid \Omega_{t-1}\Bigg)\Bigg]\leq  C/T + \epsilon.
\end{align*}

The second equation \eqref{eq:second_e} is derived by Jensen's inequality, and we show \eqref{eq:second_e0} as well as \eqref{eq:first_e0} by using  $L^r$ convergence theorem.

Next, we bound the LHS of \eqref{eq:third_e} as
\begin{align*}
&\mathrm{Cov}\Bigg(\phi_1(X_t, A_t, Y_t; \hat{g}_{t-1}, \hat{f}_{t-1}) - \phi_1(X_t, A_t, Y_t; \bar{\alpha}, f^*), \phi_2(X_t; \hat{f}_{t-1}) - \phi_2(X_t; f^*) \mid \Omega_{t-1}\Bigg)\\
&\leq \Bigg|\mathbb{E}\Bigg[ \Big(\phi_1(X_t, A_t, Y_t; \hat{g}_{t-1}, \hat{f}_{t-1}) - \phi_1(X_t, A_t, Y_t; \bar{\alpha}, f^*)\\
&\ \ \ \ \ \ \ \ \ \ \ \ \ \ \ \ \ \ \ \ \ \ \ \ \ \ \ \ \ \ \ \ \ \ \ \ \ \ \ - \mathbb{E}\left[\phi_1(X_t, A_t, Y_t; \hat{g}_{t-1}, \hat{f}_{t-1}) - \phi_1(X_t, A_t, Y_t; \bar{\alpha}, f^*) \mid \Omega_{t-1}\right]\Big)\\
&\ \ \ \ \ \ \ \ \ \ \ \ \ \ \ \ \ \ \ \times \left(\phi_2(X_t; \hat{f}_{t-1}) - \phi_2(X_t; f^*) - \mathbb{E}\left[\phi_2(X_t; \hat{f}_{t-1}) - \phi_2(X_t; f^*)\right]\right) \mid \Omega_{t-1}\Bigg]\Bigg|
\end{align*}
Then, by using the Jensen's inequality,
\begin{align*}
&\leq \mathbb{E}\Bigg[\Bigg|\Big(\phi_1(X_t, A_t, Y_t; \hat{g}_{t-1}, \hat{f}_{t-1}) - \phi_1(X_t, A_t, Y_t; \bar{\alpha}, f^*)\\
&\ \ \ \ \ \ \ \ \ \ \ \ \ \ \ \ \ \ \ \ \ \ \ \ \ \ \ \ \ \ \ \ \ \ \ \ \ \ \  - \mathbb{E}\left[\phi_1(X_t, A_t, Y_t; \hat{g}_{t-1}, \hat{f}_{t-1}) - \phi_1(X_t, A_t, Y_t; \bar{\alpha}, f^*) \mid \Omega_{t-1}\right]\Big)\\
&\ \ \ \ \ \ \ \ \ \ \ \ \ \ \ \ \ \ \ \times \left(\phi_2(X_t; \hat{f}_{t-1}) - \phi_2(X_t; f^*) - \mathbb{E}\left[\phi_2(X_t; \hat{f}_{t-1}) - \phi_2(X_t; f^*)\right]\right)\Bigg| \mid \Omega_{t-1}\Bigg]\\
&\leq C\mathbb{E}\Bigg[ \Bigg| \phi_1(X_t, A_t, Y_t; \hat{g}_{t-1}, \hat{f}_{t-1}) - \phi_1(X_t, A_t, Y_t; \bar{\alpha}, f^*)\\
&\ \ \ \ \ \ \ \ \ \ \ \ \ \ \ \ \ \ \ \ \ \ \ \ \ \ \ \ \ \ \ \ \ \ \ \ \ \ \  - \mathbb{E}\left[\phi_1(X_t, A_t, Y_t; \hat{g}_{t-1}, \hat{f}_{t-1}) - \phi_1(X_t, A_t, Y_t; \bar{\alpha}, f^*) \mid \Omega_{t-1}\right] \Bigg| \mid \Omega_{t-1}\Bigg]\\
&=\op(1),
\end{align*}
where $C>0$ is a constant. From the second to third inequality, we used consistencies of $\hat{f}_{t-1}$ and $\hat{g}_{t-1}$, which imply that for all $X_t\in\mathcal{X}$,
\begin{align*}
&\phi_1(X_t, A_t, Y_t; \hat{g}_{t-1}, \hat{f}_{t-1}) - \phi_1(X_t, A_t, Y_t; \bar{\alpha}, f^*)\\
&=\sum^K_{a=1}\left(\frac{\epol(a\mid X_t)\mathbbm{1}[A_t=a]\left(Y_t - \hat{f}_{t-1}(a, X_t)\right) }{\hat{g}_{t-1}(a\mid X_t)} - \frac{\epol(a\mid X_t)\mathbbm{1}[A_t=a]\left(Y_t - f^*(a, X_t)\right) }{\bar{\alpha}(a\mid X_t)}\right)\\
&\leq \sum^K_{a=1}\left|\frac{\epol(a\mid X_t)\mathbbm{1}[A_t=a]\left(Y_t - \hat{f}_{t-1}(a, X_t)\right) }{\hat{g}_{t-1}(a\mid X_t)} - \frac{\epol(a\mid X_t)\mathbbm{1}[A_t=a]\left(Y_t - f^*(a, X_t)\right) }{\bar{\alpha}(a\mid X_t)}\right|\\
&\leq C\sum^K_{a=1}\left|\bar{\alpha}(a\mid X_t)\left(Y_t - \hat{f}_{t-1}(a, X_t)\right) - \hat{g}_{t-1}(a\mid X_t)\left(Y_t - f^*(a, X_t)\right)\right|\\
&\leq C\sum^K_{a=1}\Big|\bar{\alpha}(a\mid X_t) - \hat{g}_{t-1}(a\mid X_t)\Big|\\
&\ \ \  - C\sum^K_{a=1}\Big|\bar{\alpha}(a\mid X_t)\hat{f}_{t-1}(a, X_t) - \hat{g}_{t-1}(a\mid X_t)\hat{f}_{t-1}(a, X_t)\\
&\ \ \ \ \ \ \ \ \ \ \ \  + \hat{g}_{t-1}(a\mid X_t)\hat{f}_{t-1}(a, X_t) - \hat{g}_{t-1}(a\mid X_t)f^*(a, X_t)\Big|\\
&\leq C\sum^K_{a=1}\Big|\bar{\alpha}(a\mid X_t) - \hat{g}_{t-1}(a\mid X_t)\Big| - C\sum^K_{a=1}\Big|\hat{f}_{t-1}(a, X_t) - f^*(a, X_t)\Big| = \op(1),
\end{align*}
where $C > 0$ is a constant.

Thus, from \eqref{eq:first_e0}--\eqref{eq:third_e0}, the variance of the bias term converges to $0$. Then, from Chebyshev's inequality,
\begin{align*}
&\mathbb{P}\Bigg(\Bigg|\sqrt{T}\frac{1}{T}\sum^T_{t=1}\Bigg\{\phi_1(X_t, A_t, Y_t; \hat{g}_{t-1}, \hat{f}_{t-1}) - \phi_1(X_t, A_t, Y_t; \bar{\alpha}, f^*)\\
&\ \ \ -\mathbb{E}\left[\phi_1(X_t, A_t, Y_t; \hat{g}_{t-1}, \hat{f}_{t-1}) - \phi_1(X_t, A_t, Y_t; \bar{\alpha}, f^*)\mid \Omega_{t-1}\right]\\
&\ \ \ + \phi_2(X_t; \hat{f}_{t-1}) - \phi_2(X_t; f^*) -\mathbb{E}\left[\phi_2(X_t; \hat{f}_{t-1}) - \phi_2(X_t; f^*)\mid \Omega_{t-1}\right]\Bigg\}\Bigg| > \varepsilon\Bigg)\\
&\leq \mathrm{Var}\Bigg(\sqrt{T}\frac{1}{T}\sum^T_{t=1}\Bigg\{\phi_1(X_t, A_t, Y_t; \hat{g}_{t-1}, \hat{f}_{t-1}) - \phi_1(X_t, A_t, Y_t; \bar{\alpha}, f^*)\\
&\ \ \ -\mathbb{E}\left[\phi_1(X_t, A_t, Y_t; \hat{g}_{t-1}, \hat{f}_{t-1}) - \phi_1(X_t, A_t, Y_t; \bar{\alpha}, f^*)\mid \Omega_{t-1}\right]\\
&\ \ \ + \phi_2(X_t; \hat{f}_{t-1}) - \phi_2(X_t; f^*) -\mathbb{E}\left[\phi_2(X_t; \hat{f}_{t-1}) - \phi_2(X_t; f^*)\mid \Omega_{t-1}\right]\Bigg\}\Bigg)/\varepsilon^2\\
&\to 0.
\end{align*}

\subsection*{Step~2: Proof of \eqref{eq:part2}}
We can calculate the LHS of \eqref{eq:part2} as
\begin{align}
&\frac{1}{T}\sum^T_{t=1}\mathbb{E}\left[\phi_1(X_t, A_t, Y_t; \hat{g}_{t-1}, \hat{f}_{t-1})\mid \Omega_{t-1}\right] + \frac{1}{T}\sum^T_{t=1}\mathbb{E}\left[\phi_2(X_t; \hat{f}_{t-1})\mid \Omega_{t-1}\right]\nonumber\\
&\ \ \ - \frac{1}{T}\sum^T_{t=1}\mathbb{E}\left[\phi_1(X_t, A_t, Y_t; \bar{\alpha}, f^*)\mid \Omega_{t-1}\right] - \frac{1}{T}\sum^T_{t=1}\mathbb{E}\left[\phi_2(X_t; f^*)\mid \Omega_{t-1}\right]\nonumber\\
& = \frac{1}{T}\sum^T_{t=1} \mathbb{E}\left[\sum^K_{a=1}\frac{\epol(a\mid X_t)\mathbbm{1}[A_t=a]\left(Y_t(a) - \hat{f}_{t-1}(a, X_t)\right) }{\hat{g}_{t-1}(a\mid X_t)}\mid \Omega_{t-1}\right] + \frac{1}{T}\sum^T_{t=1}\mathbb{E}\left[\sum^K_{a=1}\epol(a, X_t)\hat{f}_{t-1}(a, X_t)\mid \Omega_{t-1}\right]\nonumber\\
\label{eq:vanish1}
&\ \ \ - \frac{1}{T}\sum^T_{t=1}\mathbb{E}\left[\sum^K_{a=1}\frac{\epol(a\mid X_t)\mathbbm{1}[A_t=a]\left(Y_t(a) - f^*(a, X_t)\right) }{\bar{\alpha}(a\mid X_t)}\mid \Omega_{t-1}\right]\\
&\ \ \ - \frac{1}{T}\sum^T_{t=1}\mathbb{E}\left[\sum^K_{a=1}\epol(a, X_t)f^*(a, X_t)\mid \Omega_{t-1}\right]\nonumber.
\end{align}
Here, \eqref{eq:vanish1} is $0$ because 
\begin{align*}
&\frac{1}{T}\sum^T_{t=1}\mathbb{E}\left[\sum^K_{a=1}\frac{\epol(a\mid X_t)\mathbbm{1}[A_t=a]\left(Y_t(a) - f^*(a, X_t)\right) }{\bar{\alpha}(a\mid X_t)}\mid \Omega_{t-1}\right]\\
&=\frac{1}{T}\sum^T_{t=1}\mathbb{E}\left[\mathbb{E}\left[\sum^K_{a=1}\frac{\epol(a\mid X_t)\mathbbm{1}[A_t=a]\left(Y_t(a) - f^*(a, X_t)\right) }{\bar{\alpha}(a\mid X_t)}\mid X_t, \Omega_{t-1}\right]\mid \Omega_{t-1}\right]\\
&=\frac{1}{T}\sum^T_{t=1}\mathbb{E}\left[\sum^K_{a=1}\frac{\epol(a\mid X_t)\pi_t(a\mid X_t, \Omega_{t-1})}{\bar{\alpha}(a\mid X_t)}\mathbb{E}\left[f^*(a, X_t) - f^*(a, X_t)\mid X_t, \Omega_{t-1}\right]\mid \Omega_{t-1}\right].
\end{align*}
We used the law of iterated expectation, $\mathbb{E}[\mathbbm{1}[A_t=a]\mid X_t, \Omega_{t-1}] = \pi_{t}(a\mid X_t, \Omega_{t-1})$, and $\mathbb{E}[Y_t(a)\mid X_t, \Omega_{t-1}] = f^*(a, X_t)$. Therefore, we have
\begin{align*}
&\frac{1}{T}\sum^T_{t=1}\mathbb{E}\left[\phi_1(X_t, A_t, Y_t; \hat{g}_{t-1}, \hat{f}_{t-1})\mid \Omega_{t-1}\right]\\
&\ \ \  + \frac{1}{T}\sum^T_{t=1}\mathbb{E}\left[\phi_2(X_t; \hat{f}_{t-1})\mid \Omega_{t-1}\right]\nonumber\\
&\ \ \ - \frac{1}{T}\sum^T_{t=1}\mathbb{E}\left[\phi_1(X_t, A_t, Y_t; \bar{\alpha}, f^*)\mid \Omega_{t-1}\right]\\
&\ \ \ - \frac{1}{T}\sum^T_{t=1}\mathbb{E}\left[\phi_2(X_t; f^*)\mid \Omega_{t-1}\right]\nonumber\\
& = \frac{1}{T}\sum^T_{t=1} \sum^K_{a=1}\mathbb{E}\Bigg[\mathbb{E}\Bigg[\frac{\epol(a\mid X_t)\mathbbm{1}[A_t=a]\left(Y_t(a) - \hat{f}_{t-1}(a, X_t)\right) }{\hat{g}_{t-1}(a\mid X_t)}\\
&\ \ \ \ \ \ \ \ \ \ \ \ \ \ \ \ \ \ \ \ \ \ \ \  - \epol(a, X_t)\Big( f^*(a, X_t)- \hat{f}_{t-1}(a, X_t))\Big) \mid X_t, \Omega_{t-1}\Bigg]\mid \Omega_{t-1}\Bigg]\\
& = \sum^K_{a=1} \frac{1}{T}\sum^T_{t=1} \mathbb{E}\Bigg[\frac{\epol(a\mid X_t)\big(\pi_{t}(a\mid X_t, \Omega_{t-1}) - \hat{g}_{t-1}(a\mid X_t)\big)\left(f^*(a, X_t) - \hat{f}_{t-1}(a, X_t)\right) }{\hat{g}_{t-1}(a\mid X_t)}\mid \Omega_{t-1}\Bigg].
\end{align*}

From here, we drop the subscript $t$ because $X_t$ does not depend on the period in the expectation conditioned on $\Omega_{t-1}$. Then, the sum of the expectations is bounded as
\begin{align*}
&\frac{1}{T}\sum^T_{t=1} \mathbb{E}\Bigg[\frac{\epol(a\mid X)\Big(\pi_{t}(a\mid X, \Omega_{t-1}) - \hat{g}_{t-1}(a\mid X)\Big)\left(f^*(a, X) - \hat{f}_{t-1}(a, X)\right) }{\hat{g}_{t-1}(a\mid X)}  \mid \Omega_{t-1}\Bigg]\nonumber\\
&\leq \left|\frac{1}{T}\sum^T_{t=1} \mathbb{E}\Bigg[\frac{\epol(a\mid X)\Big(\pi_{t}(a\mid X, \Omega_{t-1}) - \hat{g}_{t-1}(a\mid X)\Big)\left(f^*(a, X) - \hat{f}_{t-1}(a, X)\right) }{\hat{g}_{t-1}(a\mid X)}  \mid \Omega_{t-1}\Bigg]\right|\nonumber\\
&\leq \left|\frac{1}{T}\sum^T_{t=1} \mathbb{E}\Bigg[\epol(a\mid X)\left(\frac{\pi_{t}(a\mid X, \Omega_{t-1})}{\hat{g}_{t-1}(a\mid X)} - 1\right)\left(f^*(a, X) - \hat{f}_{t-1}(a, X)\right)  \mid \Omega_{t-1}\Bigg]\right|\nonumber\\
\end{align*}
This is decomposed and bounded as
\begin{align*}
&\leq \Bigg|\frac{1}{T}\sum^T_{t=1} \mathbb{E}\Bigg[\epol(a\mid X)\left(\frac{\pi_{t}(a\mid X, \Omega_{t-1})}{\hat{g}_{t-1}(a\mid X)} - 1\right)\left(f^*(a, X) - \hat{f}_{t-1}(a, X)\right)  \mid \Omega_{t-1}\Bigg]\\
&\ \ \ \ \ \ \ \ \ - \frac{1}{T}\sum^T_{t=1} \mathbb{E}\Bigg[\epol(a\mid X)\left(\frac{\pi_{t}(a\mid X, \Omega_{t-1})}{\frac{1}{t}\sum^t_{s=1}\pi_s(a\mid X, \Omega_{s-1})} - 1\right)\left(f^*(a, X) - \hat{f}_{t-1}(a, X)\right)  \mid \Omega_{t-1}\Bigg]\\
&\ \ \ \ \ \ \ \ \ + \frac{1}{T}\sum^T_{t=1} \mathbb{E}\Bigg[\epol(a\mid X)\left(\frac{\pi_{t}(a\mid X, \Omega_{t-1})}{\frac{1}{t}\sum^t_{s=1}\pi_s(a\mid X, \Omega_{s-1})} - 1\right)\left(f^*(a, X) - \hat{f}_{t-1}(a, X)\right)  \mid \Omega_{t-1}\Bigg]\Bigg|\nonumber\\
&= \Bigg|\frac{1}{T}\sum^T_{t=1} \mathbb{E}\Bigg[\frac{\epol(a\mid X)\pi_{t}(a\mid X, \Omega_{s-1})}{\hat{g}_{t-1}(a\mid X)\frac{1}{t}\sum^t_{s=1}\pi_s(a\mid X, \Omega_{t-1})}\\
&\ \ \ \ \ \ \ \ \ \ \ \ \ \ \ \ \ \ \times \left(\frac{1}{t}\sum^t_{s=1}\pi_s(a\mid X, \Omega_{s-1}) - \hat{g}_{t-1}(a\mid X) \right)\left(f^*(a, X) - \hat{f}_{t-1}(a, X)\right)  \mid \Omega_{t-1}\Bigg]\Bigg|\\
&\ \ \ \ \ \ \ \ \ + \Bigg|\frac{1}{T}\sum^T_{t=1} \mathbb{E}\Bigg[\epol(a\mid X)\left(\frac{\pi_{t}(a\mid X, \Omega_{t-1})}{\frac{1}{t}\sum^t_{s=1}\pi_s(a\mid X, \Omega_{s-1})} - 1\right)\left(f^*(a, X) - \hat{f}_{t-1}(a, X)\right)  \mid \Omega_{t-1}\Bigg]\Bigg|\nonumber\\
\end{align*}
Then, we want to show that
\begin{align}
\label{lem1:target1}
&\Bigg|\frac{1}{T}\sum^T_{t=1} \mathbb{E}\Bigg[\frac{\epol(a\mid X)\pi_{t}(a\mid X, \Omega_{t-1})}{\hat{g}_{t-1}(a\mid X)\frac{1}{t}\sum^t_{s=1}\pi_s(a\mid X, \Omega_{s-1})}\nonumber\\
&\ \ \ \ \ \ \ \ \ \ \ \ \ \ \ \ \ \ \left(\frac{1}{t}\sum^t_{s=1}\pi_s(a\mid X, \Omega_{t-1}) - \hat{g}_{t-1}(a\mid X) \right)\left(f^*(a, X) - \hat{f}_{t-1}(a, X)\right)  \mid \Omega_{t-1}\Bigg]\Bigg|\nonumber\\
&\leq \frac{C}{T}\sum^T_{t=1} \Bigg|\mathbb{E}\Bigg[\left(\frac{1}{t}\sum^t_{s=1}\pi_s(a\mid X, \Omega_{s-1}) - \hat{g}_{t-1}(a\mid X) \right)\left(f^*(a, X) - \hat{f}_{t-1}(a, X)\right)  \mid \Omega_{t-1}\Bigg]\Bigg|\nonumber\\
& = \op(T^{-1/2})
\end{align}
and 
\begin{align}
\label{lem1:target2}
&\Bigg|\frac{1}{T}\sum^T_{t=1} \mathbb{E}\Bigg[\epol(a\mid X)\left(\frac{\pi_{t}(a\mid X, \Omega_{t-1})}{\frac{1}{t}\sum^t_{s=1}\pi_s(a\mid X, \Omega_{s-1})} - 1\right)\left(f^*(a, X) - \hat{f}_{t-1}(a, X)\right)  \mid \Omega_{t-1}\Bigg]\Bigg|\nonumber\\
& = \op(T^{-1/2}).
\end{align}
We show \eqref{lem1:target1} by using Assumption~\ref{asm:mean_stat} as
\begin{align*}
& \frac{1}{T}\sum^T_{t=1} \Bigg|\mathbb{E}\Bigg[\left(\frac{1}{t}\sum^t_{s=1}\pi_s(a\mid X, \Omega_{s-1}) - \hat{g}_{t-1}(a\mid X) \right)\left(f^*(a, X) - \hat{f}_{t-1}(a, X)\right)  \mid \Omega_{t-1}\Bigg]\Bigg|\\
& \leq \frac{1}{T}\sum^T_{t=1}\left\|\frac{1}{t}\sum^t_{s=1}\pi_{s}(a\mid X, \Omega_{s-1}) - \hat{g}_{t-1}(a\mid X)\right\|_2\left\|f^*(a, X) - \hat{f}_{t-1}(a, X)\right\|_2\\
&\leq \frac{1}{T}\sum^T_{t=1}\left(\left\|\frac{1}{t}\sum^t_{s=1}\pi_{s}(a\mid X, \Omega_{s-1}) - \bar{\alpha}(a\mid X)\right\|_2 + \left\|\bar{\alpha}(a\mid X) - \hat{g}_{t-1}(a\mid X)\right\|_2\right)\left\|f^*(a, X) - \hat{f}_{t-1}(a, X)\right\|_2\\
&= \frac{1}{T}\sum^T_{t=1}\left\|\frac{1}{t}\sum^t_{s=1}\pi_{s}(a\mid X, \Omega_{s-1}) - \bar{\alpha}(a\mid X)\right\|_2\left\|f^*(a, X) - \hat{f}_{t-1}(a, X)\right\|_2\\
&\ \ \  + \frac{1}{T}\sum^T_{t=1}\left\|\bar{\alpha}(a\mid X) - \hat{g}_{t-1}(a\mid X)\right\|_2\left\|f^*(a, X) - \hat{f}_{t-1}(a, X)\right\|_2\\
& = \frac{1}{T}\sum^T_{t=1}\op(t^{-1/2}) + \frac{1}{T}\sum^T_{t=1}\op(t^{-1/2})\\
\end{align*}

The equation \eqref{lem1:target2} holds from Assumption~\ref{asm:stationarity}. 

Then, by using the property of Riemann zeta function,
\begin{align*}
& \frac{1}{T}\sum^T_{t=1}\op(t^{-1/2}) + \frac{1}{T}\sum^T_{t=1}\op(t^{-1/2})\\
& \approx \op(T^{-1/2}) + \op(T^{-1/2})\\
& = \op(T^{-1/2}).
\end{align*}
Therefore,
\begin{align*}
&\frac{1}{T}\sum^T_{t=1}\mathbb{E}\left[\phi_1(X_t, A_t, Y_t; \hat{g}_{t-1}, \hat{f}_{t-1})\mid \Omega_{t-1}\right] + \frac{1}{T}\sum^T_{t=1}\mathbb{E}\left[\phi_2(X_t; \hat{f}_{t-1})\mid \Omega_{t-1}\right]\nonumber\\
&\ \ \ - \frac{1}{T}\sum^T_{t=1}\mathbb{E}\left[\phi_1(X_t, A_t, Y_t; \bar{\alpha}, f^*)\mid \Omega_{t-1}\right] - \frac{1}{T}\sum^T_{t=1}\mathbb{E}\left[\phi_2(X_t; f^*)\mid \Omega_{t-1}\right]= \op(T^{-1/2})
\end{align*}

\end{proof}

\section{Proof of Lemma~\ref{LEM:2}}
\label{appdx:proof:target2}

The proof procedure follows \citet{Kato2020adaptive}. 

\begin{proof}
Let $\Gamma_t(a)$ be
\begin{align*}
\Gamma_t(a; \epol) &= \frac{\epol(a\mid X_t)\mathbbm{1}[A_t = a](Y_t - f^*(a\mid X_t))}{\bar{\alpha}(a\mid X_t)} - \epol(a\mid X_t)f^*(a\mid X_t)\\
& = \frac{\epol(a\mid X_t)\mathbbm{1}[A_t = a](Y_t(a) - f^*(a\mid X_t))}{\bar{\alpha}(a\mid X_t)} - \epol(a\mid X_t)f^*(a\mid X_t).
\end{align*}
Here, we used $\mathbbm{1}[A_t = a]Y_t = \mathbbm{1}[A_t = a]\sum^K_{a=1}\mathbbm{1}[A_t = a]Y_t(a) = \mathbbm{1}[A_t = a]Y_t(a)$. Note that $\ddot{R}_T(\epol) = \frac{1}{T}\sum^T_{t=1}\sum^K_{a=1}\Gamma_t(a; \epol)$. Then, for $Z_t = \sum^K_{a=1}\Gamma_t(a; \epol) - R(\epol)$, we want to show that
\begin{align*}
\sqrt{T}\left(\ddot{R}_T(\epol) - R(\epol)\right) = \sqrt{T}\left(\frac{1}{T}\sum^T_{t=1}Z_t\right) \xrightarrow{\mathrm{d}} \mathcal{N}\left(0, \sigma^2\right).
\end{align*}

Then, the sequence $\{Z_t\}^T_{t=1}$ is an MDS; that is, 
\begin{align*}
&\mathbb{E}\big[Z_t\mid \Omega_{t-1}\big]
\\
&= \mathbb{E}\left[\sum^K_{a=1}\Gamma_t(a; \epol) - R(\epol)\mid \Omega_{t-1}\right]
\\
&= \mathbb{E}\left[\sum^K_{a=1}\epol(a\mid X_t)f^*(a, X_t) - R(\epol_t)\mid \Omega_{t-1}\right]\\
&\ \ \ \ \ \ \ \ \ \ \ \ \ \ \ \  + \mathbb{E}\left[\sum^K_{a=1}\frac{\epol(a\mid X_t)\mathbbm{1}[A_t = a](Y_t(a) - f^*(a, X_t))}{\bar{\alpha}(A_t\mid X_t)}\mid \Omega_{t-1}\right]
\\
&= 0 + \mathbb{E}\left[\mathbb{E}\left[\sum^K_{a=1}\frac{\epol(a \mid X_t)\mathbbm{1}[A_t = a](Y_t(a) - f^*(a, X_t))}{\bar{\alpha}(a\mid X_t)}\mid X_t, \Omega_{t-1}\right]\mid \Omega_{t-1}\right]\\
&= \mathbb{E}\left[\mathbb{E}\left[\sum^K_{a=1}\frac{\epol(a \mid X_t)\pi(a\mid X_t, \Omega_{t-1})(Y_t(a) - f^*(a, X_t))}{\bar{\alpha}(a\mid X_t)}\mid X_t, \Omega_{t-1}\right]\mid \Omega_{t-1}\right]\\
&= \mathbb{E}\left[\mathbb{E}\left[\sum^K_{a=1}\frac{\epol(a \mid X_t)\pi(a\mid X_t, \Omega_{t-1})(f^*(a, X_t) - f^*(a, X_t))}{\bar{\alpha}(a\mid X_t)}\mid X_t, \Omega_{t-1}\right]\mid \Omega_{t-1}\right]\\
& = 0.
\end{align*}

Therefore, to derive the asymptotic distribution, we consider applying the CLT for a martingale difference sequences (MDS) introduced in Proposition~\ref{prp:marclt}. There are  following three conditions in the statement.

\begin{description}
\item[(a)] $\mathbb{E}\big[Z^2_t\big] = \nu^2_t > 0$ with $\big(1/T\big) \sum^T_{t=1}\nu^2_t\to \nu^2 > 0$;
\item[(b)] $\mathbb{E}\big[|Z_t|^r\big] < \infty$ for some $r>2$;
\item[(c)] $\big(1/T\big)\sum^T_{t=1}Z^2_t\xrightarrow{\mathrm{p}} \nu^2$. 
\end{description}
Because we assumed the boundedness of $z_t$ by assuming the boundedness of $Y_t$, $f^*$, and $\epol/\bar{\alpha}$, the condition~(b) holds. Therefore, the remaining task is to show the conditions~(a) and (c) hold.

\subsection*{Step~1: Check of Condition~(a)}
For $\mathbb{E}\big[Z^2_t\big]$, we have
\begin{align}
&\mathbb{E}\big[Z^2_t\big]=\mathbb{E}\left[\left(\sum^K_{a=1}\left(\frac{\epol(a\mid X_t)\mathbbm{1}[A_t = a](Y_t(a) - f^*(a, X_t))}{\bar{\alpha}(a\mid X_t)} + \epol(a\mid X_t)f^*(a, X_t)\right) - R(\epol)\right)^2\right]\nonumber\\
&= \mathbb{E}\left[\left(\sum^K_{a=1}\frac{\epol(a\mid X_t)\mathbbm{1}[A_t = a](Y_t(a) - f^*(a, X_t))}{\bar{\alpha}(a\mid X_t)} + \sum^K_{a=1}\epol(a\mid X_t)f^*(a, X_t) - R(\epol)\right)^2\right]\nonumber\\
\label{eq:first_term}
&= \mathbb{E}\Bigg[\left(\sum^K_{a=1}\frac{\epol(a\mid X_t)\mathbbm{1}[A_t = a](Y_t(a) - f^*(a, X_t))}{\bar{\alpha}(a\mid X_t)}\right)^2\\
\label{eq:second_term}
&\ \ \ \ \ \ \ \ +2\left(\sum^K_{a=1}\frac{\epol(a\mid X_t)\mathbbm{1}[A_t = a](Y_t(a) - f^*(a, X_t))}{\bar{\alpha}(a\mid X_t)}\right)\left(\sum^K_{a=1}\epol(a\mid X_t)f^*(a, X_t) - R(\epol)\right)\\
&\ \ \ \ \ \ \ \ + \left(\sum^K_{a=1}\epol(a\mid X_t)f^*(a, X_t) - R(\epol)\right)^2\Bigg].\nonumber
\end{align}
For the first term \eqref{eq:first_term}, we have
\begin{align*}
&\mathbb{E}\Bigg[\left(\sum^K_{a=1}\frac{\epol(a\mid X_t)\mathbbm{1}[A_t = a](Y_t(a) - f^*(a, X_t))}{\bar{\alpha}(a\mid X_t)}\right)^2\Bigg]\\
&= \sum^K_{a=1}\mathbb{E}\Bigg[\left(\frac{\epol(a\mid X_t)\mathbbm{1}[A_t = a](Y_t(a) - f^*(a, X_t))}{\bar{\alpha}(a\mid X_t)}\right)^2\Bigg]\\
&= \sum^K_{a=1}\mathbb{E}\Bigg[\frac{\big(\epol(a\mid X_t)\big)^2\mathbbm{1}[A_t = a](Y_t(a) - f^*(a, X_t))^2}{\bar{\alpha}(a\mid X_t)}\Bigg]\\
&= \sum^K_{a=1}\mathbb{E}\Bigg[\frac{\big(\epol(a\mid X_t)\big)^2\pi_t(a\mid X_t, \Omega_{t-1})\mathrm{Var}(Y_t(a)\mid X_t)}{\bar{\alpha}^2(a\mid X_t)}\Bigg].
\end{align*}
From the first to second line, we used $\mathbbm{1}[A_t=a]\mathbbm{1}[A_t=a'] = 0$ for $a\neq a'$. From the third to fourth line, we used the conditional independence between $\mathbbm{1}[A_t = a]$ and $(Y_t - f^*(a, X_t))^2$. The second term \eqref{eq:second_term} is $0$ because
\begin{align*}
&\mathbb{E}\Bigg[\left(\sum^K_{a=1}\frac{\epol(a\mid X_t)\mathbbm{1}[A_t = a](Y_t(a) - f^*(a, X_t))}{\bar{\alpha}(a\mid X_t)}\right)\left(\sum^K_{a=1}\epol(a\mid X_t)f^*(a, X_t) - R(\epol)\right)\Bigg]\\
&=\mathbb{E}\Bigg[\left(\sum^K_{a=1}\frac{\epol(a\mid X_t)\mathbbm{1}[A_t = a](Y_t(a) - f^*(a, X_t))}{\bar{\alpha}(a\mid X_t)}\right)\left(\sum^K_{a=1}\epol(a\mid X_t)f^*(a, X_t) - R(\epol)\right)\Bigg]\\
&=\mathbb{E}\Bigg[\left(\sum^K_{a=1}\epol(a\mid X_t)f^*(a, X_t) - R(\epol)\right)\mathbb{E}\Bigg[\sum^K_{a=1}\frac{\epol(a\mid X_t)\mathbbm{1}[A_t = a](Y_t(a) - f^*(a, X_t))}{\bar{\alpha}(a\mid X_t)}\mid X_t, \Omega_{t-1}\Bigg]\Bigg]\\
&=\mathbb{E}\Bigg[\left(\sum^K_{a=1}\epol(a\mid X_t)f^*(a, X_t) - R(\epol)\right)\Bigg(\sum^K_{a=1}\frac{\epol(a\mid X_t)\pi_t(a\mid X_t, \Omega_{t-1})(f^*(a, X_t) - f^*(a, X_t))}{\bar{\alpha}(a\mid X_t)}\Bigg)\Bigg]\\
&=0.
\end{align*}
In conclusion, we have
\begin{align*}
&\mathbb{E}\big[Z^2_t\big] = \mathbb{E}\Bigg[\sum^K_{a=1}\frac{\big(\epol(a\mid X_t)\big)^2\pi_t(a\mid X_t, \Omega_{t-1})\mathrm{Var}(Y_t(a)\mid X_t)}{\bar{\alpha}^2(a\mid X_t)} + \left(\sum^K_{a=1}\epol(a\mid X_t)f^*(a, X_t) - R(\epol)\right)^2\Bigg].
\end{align*}
Because $X_t$ and $Y_t(a)$ does not depend on the period in the expectation, by dropping their subscripts, we represent $\nu^2_t = \mathbb{E}\big[Z^2_t\big]$ as 
\begin{align*}
&\nu^2_t = \mathbb{E}\Bigg[\sum^K_{a=1}\frac{\big(\epol(a\mid X)\big)^2\pi_t(a\mid X, \Omega_{t-1})\mathrm{Var}(Y(a)\mid X)}{\bar{\alpha}^2(a\mid X)} + \left(\sum^K_{a=1}\epol(a\mid X)f^*(a, X) - R(\epol)\right)^2\Bigg].
\end{align*}
Next, we show that for $\nu^2 = \mathbb{E}\Bigg[\sum^K_{a=1}\frac{\big(\epol(a\mid X)\big)^2\bar{\alpha}(a\mid X)\mathrm{Var}(Y(a)\mid X)}{\bar{\alpha}^2(a\mid X)} + \left(\sum^K_{a=1}\epol(a\mid X)f^*(a, X) - R(\epol)\right)^2\Bigg]$, 
\begin{align*}
&\frac{1}{T}\sum^T_{t=1}\nu^2_t - \nu^2 \to 0\ \ \ \mathrm{as}\ T\to\infty\\
&\Leftrightarrow \frac{1}{T}\sum^T_{t=1}\mathbb{E}\Bigg[\sum^K_{a=1}\frac{\big(\epol(a\mid X)\big)^2\pi_t(a\mid X, \Omega_{t-1})\mathrm{Var}(Y(a)\mid X)}{\bar{\alpha}^2(a\mid X)} + \left(\sum^K_{a=1}\epol(a\mid X)f^*(a, X) - R(\epol)\right)^2\Bigg]\\
&\ \ \ - \mathbb{E}\Bigg[\sum^K_{a=1}\frac{\big(\epol(a\mid X)\big)^2\bar{\alpha}(a\mid X)\mathrm{Var}(Y(a)\mid X)}{\bar{\alpha}^2(a\mid X)} + \left(\sum^K_{a=1}\epol(a\mid X)f^*(a, X) - R(\epol)\right)^2\Bigg]\\
&\Leftrightarrow \sum^K_{a=1}\mathbb{E}\Bigg[\frac{1}{T}\sum^T_{t=1}\frac{\big(\epol(a\mid X)\big)^2\pi_t(a\mid X, \Omega_{t-1})\mathrm{Var}(Y(a)\mid X)}{\bar{\alpha}^2(a\mid X)} - \frac{\big(\epol(a\mid X)\big)^2\bar{\alpha}(a\mid X)\mathrm{Var}(Y(a)\mid X)}{\bar{\alpha}^2(a\mid X)}\Bigg]\\
&\ \ \ \ \ \ \ \ \ \ \ \ \ \ \ \ \ \ \to 0\\
&\Leftrightarrow \sum^K_{a=1}\mathbb{E}\Bigg[\frac{\big(\epol(a\mid X)\big)^2\mathrm{Var}(Y(a)\mid X)}{\bar{\alpha}^2(a\mid X)}\left(\frac{1}{T}\sum^T_{t=1}\pi_t(a\mid X, \Omega_{t-1}) - \bar{\alpha}(a\mid X)\right)\Bigg]\to 0.
\end{align*}
Because $\frac{\big(\epol(a\mid X)\big)^2\mathrm{Var}(Y(a)\mid X)}{\bar{\alpha}^2(a\mid X)}$ is upper bounded, there is a constant $C>0$ such that
\begin{align*}
&\mathbb{E}\Bigg[\frac{\big(\epol(a\mid X)\big)^2\mathrm{Var}(Y(a)\mid X)}{\bar{\alpha}^2(a\mid X)}\left(\frac{1}{T}\sum^T_{t=1}\pi_t(a\mid X, \Omega_{t-1}) - \bar{\alpha}(a\mid X)\right)\Bigg]\\
&\leq C \mathbb{E}\Bigg[\left|\frac{1}{T}\sum^T_{t=1}\pi_t(a\mid X, \Omega_{t-1}) - \bar{\alpha}(a\mid X)\right|\Bigg].
\end{align*}

We assumed that the  point-wise convergence of $\frac{1}{T}\sum^T_{t=1}\pi_t(a\mid x, \Omega_{t-1})$; that is, for all $x\in\mathcal{X}$, $k\in\mathcal{A}$, and $\Omega_{t-1} \in \mathcal{M}_{t-1}$, $\frac{1}{T}\sum^T_{t=1}\pi_t(a\mid x, \Omega_{t-1})\xrightarrow{\mathrm{d}} \bar{\alpha}(a\mid x)$.  From this assumption, if $\frac{1}{T}\sum^T_{t=1}\pi_t(a\mid x, \Omega_{t-1})$ is  uniformly integrable, we can show that 
\begin{align*}
&\mathbb{E}\left[\left|\frac{1}{T}\sum^T_{t=1}\pi_t(a\mid X, \Omega_{t-1}) - \bar{\alpha}(a\mid x)\right| \mid X=x\right] = \mathbb{E}\left[\left|\frac{1}{T}\sum^T_{t=1}\pi_t(a\mid x, \Omega_{t-1}) - \bar{\alpha}(a\mid x)\right| \right]  \to 0,
\end{align*}
as $t\to\infty$ using $L^r$-convergence theorem (Proposition~\ref{prp:lr_conv_theorem}). Note that $X_t$ is independent from $\Omega_{t-1}$. Here, for a fixed $x$, we can show that $\frac{1}{T}\sum^T_{t=1}\pi_t(a\mid X, \Omega_{t-1})$ is uniformly integrable from the boundedness of $\frac{1}{T}\sum^T_{t=1}\pi_t(a\mid X, \Omega_{t-1})$ (Proposition~\ref{prp:suff_uniint}). From the point-wise convergence of $\mathbb{E}\left[\left|\frac{1}{T}\sum^T_{t=1}\pi_t(a\mid x, \Omega_{t-1}) - \bar{\alpha}(a\mid x)\right| \mid X=x\right]$, by using the Lebesgue's dominated convergence theorem, we can show that

\begin{align*}
&\mathbb{E}\left[\mathbb{E}\left[\left|\frac{1}{T}\sum^T_{t=1}\pi_t(a\mid X, \Omega_{t-1}) - \bar{\alpha}(a\mid X)\right| \mid X\right]\right]  \to 0.
\end{align*}
In conclusion, as $t\to\infty$,

\begin{align*}
&\mathbb{E}\big[Z^2_t\big] - \mathbb{E}\Bigg[\sum^K_{a=1}\frac{\big(\epol(a\mid X)\big)^2\bar{\alpha}(a\mid X)\mathrm{Var}(Y(a)\mid X)}{\bar{\alpha}^2(a\mid X)} + \left(\sum^K_{a=1}\epol(a\mid X)f^*(a, X) - R(\epol)\right)^2\Bigg]\to 0.
\end{align*}

\subsection*{Step~2: Check of Condition~(c)}
Let $U_t$ be an MDS such that

\begin{align*}
&U_t = Z^2_t - \mathbb{E}\big[Z^2_t\mid \Omega_{t-1}\big]\\
&=\left(\sum^K_{a=1}\left(\frac{\epol(a\mid X_t)\mathbbm{1}[A_t = a](Y_t(a) - f^*(a, X_t))}{\bar{\alpha}(a\mid X_t)} + \epol(a\mid X_t)f^*(a, X_t)\right) - R(\epol)\right)^2\\\
& - \mathbb{E}\left[\left(\sum^K_{a=1}\left(\frac{\epol(a\mid X_t)\mathbbm{1}[A_t = a](Y_t(a) - f^*(a, X_t))}{\bar{\alpha}(a\mid X_t)} + \epol(a\mid X_t)f^*(a, X_t)\right) - R(\epol)\right)^2\mid \Omega_{t-1}\right].
\end{align*}
From the boundedness of each variable in $Z_t$, we can apply weak law of large numbers for an MDS (Proposition~\ref{prp:mrtgl_WLLN} in Appendix~\ref{appdx:prelim}). Then, we have

\begin{align*}
&\frac{1}{T}\sum^T_{t=1}U_t = \frac{1}{T}\sum^T_{t=1}\big(Z^2_t - \mathbb{E}\big[Z^2_t\mid \Omega_{t-1}\big]\big)\xrightarrow{\mathrm{p}} 0.
\end{align*}
Next, we show that

\begin{align*}
\frac{1}{T}\sum^T_{t=1}\mathbb{E}\big[Z^2_t\mid \Omega_{t-1}\big] - \sigma^2\xrightarrow{\mathrm{p}} 0.
\end{align*}
From Markov's inequality, for any $\varepsilon > 0$, we have

\begin{align*}
&\mathbb{P}\left(\left|\frac{1}{T}\sum^T_{t=1}\mathbb{E}\big[Z^2_t\mid \Omega_{t-1}\big] - \sigma^2\right| \geq \varepsilon\right)\leq \frac{\mathbb{E}\left[\left|\frac{1}{T}\sum^T_{t=1}\mathbb{E}\big[Z^2_t\mid \Omega_{t-1}\big] - \sigma^2\right|\right]}{\varepsilon}.
\end{align*}
Then, we consider showing $\mathbb{E}\left[\left|\frac{1}{T}\sum^T_{t=1}\mathbb{E}\big[Z^2_t\mid \Omega_{t-1}\big] - \sigma^2\right|\right] \to 0$. As well as Step~1, we have
\begin{align*}
&\mathbb{E}\big[Z^2_t\mid \Omega_{t-1}\big] = \\
&\mathbb{E}\Bigg[\sum^K_{a=1}\frac{\big(\epol(a\mid X)\big)^2\pi_t(a\mid X, \Omega_{t-1})\mathrm{Var}(Y(a)\mid X)}{\bar{\alpha}^2(a\mid X)} + \left(\sum^K_{a=1}\epol(a\mid X_t)f^*(a, X) - R(\epol)\right)^2\mid \Omega_{t-1}\Bigg].
\end{align*}
As Step~1, we drop the subscript $t$ from $X_t$ and $Y_t(a)$. Then,

\begin{align*}
&\mathbb{E}\left[\left|\frac{1}{T}\sum^T_{t=1}\mathbb{E}\big[Z^2_t\mid \Omega_{t-1}\big] - \sigma^2\right|\right]\\
&=\mathbb{E}\Bigg[\Bigg|\frac{1}{T}\sum^T_{t=1}\mathbb{E}\Bigg[\sum^K_{a=1}\frac{\big(\epol(a\mid X)\big)^2\pi_t(a\mid X, \Omega_{t-1})\mathrm{Var}(Y(a)\mid X)}{\bar{\alpha}^2(a\mid X)}\\
&\ \ \ \ \ \ \ \ \ \ \ \ \ \ \ \ \ \ \ \ \ \ \ \ \ \ \ \ \ \ \ \ \ \ \ \ \ \ \ \ \ \ \ \ \ \ \ \  + \left(\sum^K_{a=1}\epol(a\mid X)f^*(a, X) - R(\epol)\right)^2\mid \Omega_{t-1}\Bigg]\\
&\ \ \ - \mathbb{E}\Bigg[\sum^K_{a=1}\frac{\big(\epol(a\mid X)\big)^2\bar{\alpha}(a\mid X)\mathrm{Var}(Y(a)\mid X)}{\bar{\alpha}^2(a\mid X)} + \left(\sum^K_{a=1}\epol(a\mid X)f^*(a, X) - R(\epol)\right)^2\Bigg]\Bigg|\Bigg]\\
&=\mathbb{E}\Bigg[\Bigg|\frac{1}{T}\sum^T_{t=1}\mathbb{E}\Bigg[\sum^K_{a=1}\frac{\big(\epol(a\mid X)\big)^2\pi_t(a\mid X, \Omega_{t-1})\mathrm{Var}(Y(a)\mid X)}{\bar{\alpha}^2(a\mid X)}\mid \Omega_{t-1}\Bigg]\\
&\ \ \ - \mathbb{E}\Bigg[\sum^K_{a=1}\frac{\big(\epol(a\mid X)\big)^2\bar{\alpha}(a\mid X)\mathrm{Var}(Y(a)\mid X)}{\bar{\alpha}^2(a\mid X)} \Bigg]\Bigg|\Bigg]\\
\end{align*}

\begin{align*}
&=\mathbb{E}\Bigg[\Bigg|\frac{1}{T}\sum^T_{t=1}\mathbb{E}\Bigg[\sum^K_{a=1}\frac{\big(\epol(a\mid X)\big)^2\pi_t(a\mid X, \Omega_{t-1})\mathrm{Var}(Y(a)\mid X)}{\bar{\alpha}^2(a\mid X)}\mid \Omega_{t-1}\Bigg]\\
&\ \ \ - \mathbb{E}\Bigg[\sum^K_{a=1}\frac{\big(\epol(a\mid X)\big)^2\bar{\alpha}(a\mid X)\mathrm{Var}(Y(a)\mid X)}{\bar{\alpha}^2(a\mid X)} \mid \Omega_{t-1} \Bigg]\Bigg|\Bigg]\\
&=\mathbb{E}\Bigg[\Bigg|\frac{1}{T}\sum^T_{t=1}\mathbb{E}\Bigg[\sum^K_{a=1}\frac{\big(\epol(a\mid X)\big)^2\mathrm{Var}(Y(a)\mid X)}{\bar{\alpha}^2(a\mid X)}\Big(\pi_t(a\mid X, \Omega_{t-1}) - \bar{\alpha}(a\mid X)\Big) \mid \Omega_{t-1}\Bigg] \Bigg|\Bigg]\\
&\leq\sum^K_{a=1}\mathbb{E}\Bigg[\Bigg|\frac{1}{T}\sum^T_{t=1}\mathbb{E}\Bigg[\frac{\big(\epol(a\mid X)\big)^2\mathrm{Var}(Y(a)\mid X)}{\bar{\alpha}^2(a\mid X)}\Big(\pi_t(a\mid X, \Omega_{t-1}) - \bar{\alpha}(a\mid X)\Big) \mid \Omega_{t-1}\Bigg] \Bigg|\Bigg]\\
&=\sum^K_{a=1}\mathbb{E}\Bigg[\Bigg|\mathbb{E}\Bigg[\frac{\big(\epol(a\mid X)\big)^2\mathrm{Var}(Y(a)\mid X)}{\bar{\alpha}^2(a\mid X)}\frac{1}{T}\sum^T_{t=1}\Big(\pi_t(a\mid X, \Omega_{t-1}) - \bar{\alpha}(a\mid X)\Big) \mid \Omega_{T-1}\Bigg] \Bigg|\Bigg]\\
&=\sum^K_{a=1}\mathbb{E}\Bigg[\Bigg|\mathbb{E}\Bigg[\frac{\big(\epol(a\mid X)\big)^2\mathrm{Var}(Y(a)\mid X)}{\bar{\alpha}^2(a\mid X)}\mathbb{E}\Bigg[\frac{1}{T}\sum^T_{t=1}\Big(\pi_t(a\mid X, \Omega_{t-1}) - \bar{\alpha}(a\mid X)\Big) \mid X, \Omega_{T-1}\Bigg] \mid \Omega_{T-1}\Bigg] \Bigg|\Bigg]\\
&=\sum^K_{a=1}\mathbb{E}\Bigg[\Bigg|\mathbb{E}\Bigg[\frac{\big(\epol(a\mid X)\big)^2\mathrm{Var}(Y(a)\mid X)}{\bar{\alpha}^2(a\mid X)}\mathbb{E}\Bigg[\frac{1}{T}\sum^T_{t=1}\pi_t(a\mid X, \Omega_{t-1}) - \bar{\alpha}(a\mid X) \mid X, \Omega_{T-1}\Bigg] \mid \Omega_{T-1}\Bigg] \Bigg|\Bigg].
\end{align*}
Then, by using Jensen's inequality, 

\begin{align*}
&\mathbb{E}\left[\big|\mathbb{E}\big[Z^2_t\mid \Omega_{t-1}\big] - \sigma^2\big|\right]\\
&=\sum^K_{a=1}\mathbb{E}\Bigg[\mathbb{E}\Bigg[\Bigg|\frac{\big(\epol(a\mid X)\big)^2\mathrm{Var}(Y(a)\mid X)}{\bar{\alpha}^2(a\mid X)}\mathbb{E}\Bigg[\frac{1}{T}\sum^T_{t=1}\pi_t(a\mid X, \Omega_{t-1}) - \bar{\alpha}(a\mid X) \mid X_t, \Omega_{T-1}\Bigg]\Bigg| \mid \Omega_{T-1}\Bigg] \Bigg]\\
&=\sum^K_{a=1}\mathbb{E}\Bigg[\frac{\big(\epol(a\mid X)\big)^2\mathrm{Var}(Y(a)\mid X)}{\bar{\alpha}^2(a\mid X)}\Bigg|\mathbb{E}\Bigg[\frac{1}{T}\sum^T_{t=1}\pi_t(a\mid X, \Omega_{t-1}) - \bar{\alpha}(a\mid X) \mid X_t, \Omega_{T-1}\Bigg]\Bigg|  \Bigg].
\end{align*}

Then, from $L^r$ convergence theorem, by using pointwise convergence of $\frac{1}{T}\sum^T_{t=1}\pi(a\mid x, \Omega_{t-1})$ and the boundedness of $Z_t$, we have $\frac{\mathbb{E}\left[\left|\frac{1}{T}\sum^T_{t=1}\mathbb{E}\big[Z^2_t\mid \Omega_{t-1}\big] - \sigma^2\right|\right]}{\varepsilon} \to 0$. Therefore,

\begin{align*}
&\mathbb{P}\left(\left|\frac{1}{T}\sum^T_{t=1}\mathbb{E}\big[Z^2_t\mid \Omega_{t-1}\big] - \sigma^2\right| \geq \varepsilon\right) \leq \frac{\mathbb{E}\left[\left|\frac{1}{T}\sum^T_{t=1}\mathbb{E}\big[Z^2_t\mid \Omega_{t-1}\big] - \sigma^2\right|\right]}{\varepsilon} \to 0.
\end{align*}
In conclusion, 

\begin{align*}
&\frac{1}{T}\sum^T_{t=1}Z^2_t - \sigma^2 = \frac{1}{T}\sum^T_{t=1}\big(Z^2_t - \mathbb{E}\left[Z^2_t\mid \Omega_{t-1}\big] + \mathbb{E}\big[Z^2_t\mid \Omega_{t-1}\big] - \sigma^2\right)\xrightarrow{\mathrm{p}} 0.
\end{align*}

\subsection*{Step~3: Conclusion}
We can use CLT for an MDS. Hence, we have

\begin{align*} 
\sqrt{T}\left(\tilde{R}\left(\epol\right) - R(\epol)\right) \to \mathcal{N}\left(0, \sigma^2\right),
\end{align*}
where $\sigma^2 = \mathbb{E}\Bigg[\sum^K_{a=1}\frac{\big(\epol(a\mid X)\big)^2\mathrm{Var}(Y(a)\mid X)}{\bar{\alpha}(a\mid X)} + \left(\sum^K_{a=1}\epol(a\mid X)f^*(a, X) - R(\epol)\right)^2\Bigg]$.
\end{proof}

\end{document}

%% file: def.tex
\usepackage{amsfonts,amssymb,graphics,epsfig,verbatim,bm,bbm,latexsym,amsmath,url,amsbsy,multicol,multirow,wrapfig,natbib}

\newtheorem{theorem}{Theorem}
\newtheorem{assumption}{Assumption}
\newtheorem{proposition}{Proposition}

\newtheorem{lemma}{Lemma}

\newtheorem{remark}{Remark}
\newtheorem{definition}{Definition}

\def\Holder{{H\"{o}lder}}

\def\Cramer{Cram\'{e}r}

\newcommand{\Op}{\mathrm{O}_{p}}
\newcommand{\op}{\mathrm{o}_{p}}

\newcommand{\pa}{\mathrm{\pa}}

\newcommand{\epol}{\pi^\mathrm{e}}

\renewcommand{\eqref}[1]{(\ref{#1})}

\newcommand{\RN}[1]{%
  \textup{\uppercase\expandafter{\romannumeral#1}}%
}